\documentclass[3p,12pt]{elsarticle}

\usepackage{amsmath,amssymb,mathtools}
\usepackage{amsthm}
\usepackage{mdframed}
\usepackage{bm}
\usepackage{graphicx}
\usepackage{booktabs,tabularx}
\usepackage{placeins}
\usepackage{hyperref}
\usepackage{tikz}
\usetikzlibrary{arrows.meta,positioning,calc,fit,shapes.geometric}

\usepackage{natbib}

\usepackage[english]{babel}
\usepackage{lineno}

\usepackage{amsfonts}

\newtheorem{theorem}{Theorem}[section]
\newtheorem{lemma}[theorem]{Lemma}
\newtheorem{proposition}[theorem]{Proposition}
\newtheorem{corollary}[theorem]{Corollary}

\theoremstyle{definition}
\newtheorem{definition}[theorem]{Definition}

\theoremstyle{remark}
\newtheorem{remark}[theorem]{Remark}

\newcolumntype{Y}{>{\raggedright\arraybackslash}X}

\begin{document}
\title{Path-Dependent Entropic Lagrangian for Probability Flows: Balance--Entropy Routing and Composable Information Potentials}

\author[TU]{Huilong Ren}
\ead{hlren@tongji.edu.cn}
\address[TU]{State Key Laboratory of Disaster Reduction in Civil Engineering, College of Civil Engineering, Tongji University, Shanghai 200092, China}

\begin{abstract}
Probability distributions are central to information theory, statistical inference, and modern probabilistic learning. Maximum entropy selects a probability state under prescribed constraints, but it does not specify how that state is reached, how probability is transported, or how dissipation and external information exchange are accounted for along the path. We develop a path-dependent entropic Lagrangian calculus that extends static state selection to probability-path evolution through restricted generators, upper-limit history terms, and explicit balance--entropy port routing. The construction yields the thermal state relation, conservative probability balance, and nonnegative production under standard mobility closure. Its KL/Shannon sector recovers maximum-entropy and Bayesian laws as stationary no-flux states, while time-dependent information potentials separate internal dissipation from supplied information power. Composable information and structural potentials control tails, sparsity, robustness, regularization, and nonlocal multimodality without changing the accounting architecture. Two numerical examples verify mass conservation, energy decomposition, and the total free-energy ledger.

\end{abstract}

\maketitle

\section{Introduction}\label{sec:introduction}

Probability distributions are a common mathematical state description for information theory, statistical inference, and much of modern probabilistic learning. Entropy, relative entropy, cross-entropy, and mutual information are functionals of probability laws \cite{shannon1948,kullback1951,cover2006}; variational inference optimizes distribution-valued objectives \cite{kingma2014auto,blei2017variational}; and generative methods construct transformations between probability laws \cite{sohl2015deep,ho2020denoising,song2021scorebased}. These settings require not only principles for selecting probability states, but also laws for their admissible evolution.

Maximum entropy provides a canonical static selection rule under prescribed information constraints \cite{jaynes1957,cover2006}. A positive target density can be represented in Gibbs form through a suitable information potential, but the static rule does not determine the transition path, probability balance, irreversible production, or open-system information exchange. Bayesian updating likewise specifies a posterior endpoint while leaving its dynamical realization as a separate question \cite{mackay2003,bernardo2009}.

Thermodynamics supplies the missing language of stored energy, flux--force relations, entropy production, and boundary exchange \cite{onsager1931reciprocal,degroot1984non,seifert2012stochastic,esposito2010three}. Related structures occur in GENERIC-type formulations and in free-energy or Wasserstein descriptions of Fokker--Planck dynamics \cite{grmela1997dynamics,ottinger2005gen,jordan1998variational,otto2001geometry,ambrosio2008gradient}. The present work develops a different history-channel construction in which upper-limit variations and explicit port routing separate probability balance from entropy accounting.

The Path-Dependent Entropic Lagrangian (PDEL) \cite{ren2025thermomechanical,ren2026variational} is an energy-valued, history-aware functional containing stored energy, thermal pairing, accumulated channel expenditure, and boundary or information ports. Restricted generators produce the thermal state relation, conservative probability balance, and nonnegative production under standard mobility closure. Maximum entropy is recovered as the stationary closed-system sector: the KL/Shannon module gives MaxEnt and Bayesian distributions as no-flux states, while a time-dependent information potential separates internal dissipation from externally supplied information power.

The framework also permits constitutive substitution. Generalized entropic charts control heavy tails, sparsity, and robustness \cite{tsallis1988,borland1998,naudts2011,basu1998robust}; local structural energies regulate shape \cite{carlen1991,rudin1992nonlinear,chambolle2004,chan2000higher}; and nonlocal terms support clustering and multimodality \cite{carrillo2003asymptotic,carrillo2010particle}. These modules change the model while preserving the same balance--entropy architecture.

The paper proceeds as follows. Section~2 gives a concise thermomechanical origin of the channel rules. Sections~3--5 formulate the probability-path calculus and its KL/Shannon baseline. Section~6 develops composable information and structural modules. Section~7 connects maximum-entropy states, Bayesian updating, and information-driven paths. Section~8 presents two numerical ledger checks immediately before the conclusions.

\section{Thermomechanical review and thermoelastic coupling}\label{sec:thermomechanical-review}

Path-dependent thermomechanical formulations show how one history-aware functional can generate mechanical balance and the entropy/heat equation while assigning a dissipative power to both its resistance and thermal roles \cite{ren2025thermomechanical,ren2026variational}. The following prototype records only the ingredients used later for probability flows.

\subsection{A minimal thermoelastic--damping prototype}\label{sec:thermoelastic-prototype}

Let $\bm u$ be the displacement on $\Omega\subset\mathbb R^d$, with small strain
\begin{equation}\label{eq:thermo-strain}
\bm\varepsilon(\bm u)=\frac12\big(\nabla\bm u+\nabla\bm u^{\mathsf T}\big),
\qquad
\bm\varepsilon^{\mathrm e}=\bm\varepsilon-\bm\alpha(\theta-\theta_0).
\end{equation}
Heat conduction obeys $\bm q=-k\nabla\theta$, $k\ge0$, and the Helmholtz free energy is
\begin{align}
\phi(\bm\varepsilon,\theta)
&=\frac12\bm\varepsilon^{\mathrm e}:\mathbb C:\bm\varepsilon^{\mathrm e}
+\rho c_v\left[(\theta-\theta_0)-\theta\log\left(\frac{\theta}{\theta_0}\right)\right],
\label{eq:thermo-free-energy}\\
\bm\sigma&=\partial_{\bm\varepsilon}\phi=\mathbb C:\bm\varepsilon^{\mathrm e}.
\label{eq:thermo-stress}
\end{align}
For the damping power $D_{\mathrm d}=c_{\mathrm d}|\dot{\bm u}|^2$, $c_{\mathrm d}\ge0$, consider
\begin{equation}\label{eq:thermo-pdel-action}
\mathcal L_{\mathrm{thm}}
=\int_{t_0}^{T}\!\int_{\Omega}
\left[
\frac12\rho|\dot{\bm u}|^2-\phi-\theta s
-\int_{t_0}^{t}\left(\nabla\cdot\bm q-D_{\mathrm d}\right)d\tau
\right]dV\,dt.
\end{equation}

\subsection{Upper-limit channels and coupled equations}\label{sec:thermoelastic-variation}

The history integrand is read only at its current upper limit. The increment $\delta t_1$ denotes the entropy/thermal channel and $\delta t_2$ the mechanical channel; they are two readings of the same path endpoint, not two physical times. Using a single variation symbol,
\begin{align}
\delta\left(\int_{t_0}^{t}\nabla\cdot\bm q\,d\tau\right)
&=\nabla\cdot\bm q\,\delta t_1,
\label{eq:thermo-heat-channel}\\
\delta\left(\int_{t_0}^{t}D_{\mathrm d}\,d\tau\right)
&=D_{\mathrm d}\,\delta t_1-D_{\mathrm d}\,\delta t_2,
\qquad
\delta\bm u=\dot{\bm u}\,\delta t_2.
\label{eq:thermo-spending-rule}
\end{align}
Thus heat divergence is a pure entropy-channel term, while the damping power produces both irreversible heating and the mechanical resistance $c_{\mathrm d}\dot{\bm u}\cdot\delta\bm u$.

After the usual integrations by parts, the variation is
\begin{align}
\delta\mathcal L_{\mathrm{thm}}
=\int_{t_0}^{T}\!\int_{\Omega}
\Big[&-\big(\rho\ddot{\bm u}-\nabla\cdot\bm\sigma+c_{\mathrm d}\dot{\bm u}\big)\cdot\delta\bm u
-\big(\partial_\theta\phi+s\big)\delta\theta \notag\\
&-\big(\theta\dot s+\nabla\cdot\bm q-c_{\mathrm d}|\dot{\bm u}|^2\big)\delta t_1
\Big]dV\,dt.
\label{eq:thermo-total-variation}
\end{align}
Stationarity gives
\begin{align}
\rho\ddot{\bm u}-\nabla\cdot\bm\sigma+c_{\mathrm d}\dot{\bm u}&=0,
\label{eq:thermo-momentum}\\
s&=-\partial_\theta\phi,
\label{eq:thermo-conjugacy}\\
\theta\dot s+\nabla\cdot\bm q&=c_{\mathrm d}|\dot{\bm u}|^2.
\label{eq:thermo-entropy-balance}
\end{align}
For Eq.~\ref{eq:thermo-free-energy},
\begin{equation}\label{eq:thermo-entropy-explicit}
s=\bm\alpha:\bm\sigma+\rho c_v\log\left(\frac{\theta}{\theta_0}\right),
\end{equation}
and the temperature equation becomes
\begin{equation}\label{eq:thermo-temperature-form}
\rho c_v\dot\theta
+\theta\,\bm\alpha:\mathbb C:\big(\dot{\bm\varepsilon}-\bm\alpha\dot\theta\big)
-\nabla\cdot(k\nabla\theta)
=c_{\mathrm d}|\dot{\bm u}|^2.
\end{equation}
The reversible thermoelastic term transfers stored energy between mechanical and thermal fields, whereas damping supplies a nonnegative heat source.

\subsection{Rules transferred to probability flows}\label{sec:thermo-three-rules}

Three features are retained below: thermal conjugacy follows from the $\delta\theta$ coefficient; upper-limit history terms supply their present channel power; and divergence ports are split pointwise before assignment. For a potential $\mu$ and flux $j$,
\begin{equation}\label{eq:thermo-port-split-review}
-\nabla\cdot(\mu j)=-\mu\,\nabla\cdot j-j\cdot\nabla\mu.
\end{equation}
The first term enters the species or probability-balance channel, while the second enters the entropy channel. Under $j=-M\nabla\mu$ with $M\succeq0$,
\begin{equation}\label{eq:thermo-port-production-review}
-j\cdot\nabla\mu=\nabla\mu\cdot M\nabla\mu\ge0.
\end{equation}
This is the direct bridge from thermoelastic channel accounting to the probability-flow construction.

\section{Path-dependent entropic Lagrangian and admissible accounting rules}\label{sec:ledger-admissible}

\subsection{Setting and notation}\label{sec:setting-notation}

Let $X\subset\mathbb{R}^d$ be a fixed state domain with (piecewise) Lipschitz boundary $\partial X$ and outward unit normal $n$. Time is denoted by $t\in[t_0,T]$. We consider a probability density $p(\cdot,t):X\to[0,\infty)$ satisfying the initial normalization
\begin{equation}\label{eq:mass-one-initial}
\int_X p(x,t_0)\,dx=1.
\end{equation}
The probability flux is denoted by $j(\cdot,t):X\to\mathbb{R}^d$, and the balance (continuity) equation is written in conservative form as
\begin{equation}\label{eq:continuity}
\partial_t p+\nabla\cdot j=0\qquad\text{in }X\times(t_0,T].
\end{equation}
A standard boundary closure is the no-flux condition
\begin{equation}\label{eq:noflux}
j\cdot n=0\qquad\text{on }\partial X\times(t_0,T],
\end{equation}
which yields conservation of total probability. Indeed, integrating Eq.~\ref{eq:continuity} over $X$ and applying the divergence theorem gives
\begin{equation}\label{eq:mass-conservation}
\frac{d}{dt}\int_X p(x,t)\,dx=-\int_{\partial X} j\cdot n\,dA,
\end{equation}
hence Eq.~\ref{eq:noflux} implies $\int_X p(x,t)\,dx=\int_X p(x,t_0)\,dx=1$ for all $t\in[t_0,T]$. When $X=\mathbb{R}^d$, the role of Eq.~\ref{eq:noflux} is played by a decay condition ensuring that the boundary integral at infinity vanishes. The conservative representation Eq.~\ref{eq:continuity} and the associated mass conservation are standard in probabilistic transport and Fokker--Planck-type dynamics \cite{ambrosio2008gradient,risken1989fokker}.

A strictly positive reference density (base measure) is denoted by $\pi:X\to(0,\infty)$, which may be taken as a constant (uniform base measure on bounded domains) or as a prescribed prior. We introduce the relative density
\begin{equation}\label{eq:relative-density}
u:=\frac{p}{\pi}.
\end{equation}
Thermal variables are represented by the absolute temperature $\theta(\cdot,t):X\to(0,\infty)$ and an entropy density $s(\cdot,t):X\to\mathbb{R}$. If thermal transport is retained, the heat flux is denoted by $q(\cdot,t):X\to\mathbb{R}^d$. Throughout the paper $q$ denotes the physical conductive heat flux, positive in the direction of heat transport across outward oriented surfaces; the default Fourier closure is
\begin{equation}\label{eq:fourier-q}
q=-k\nabla\theta,\qquad k\succeq 0,
\end{equation}
with $k>0$ in the scalar isotropic case. The heat channel enters the entropy tracking only through $\nabla\cdot q$, consistent with the channel-routing rules specified later.


\subsection{Path-dependent entropic Lagrangian functional and derivation rules}

\begin{definition}[Probability path and entropic Lagrangian]\label{def:probability-path}
A probability path is a curve $t\mapsto p_t\in\mathcal P(X)$ together with the fluxes and ports that account for its transport and exchange. An entropic Lagrangian couples the current stored state to accumulated history terms whose upper-limit contributions are routed into balance and entropy channels. It is path-dependent through this history contribution and entropic through its information potentials, thermal pairing, and irreversible production; physical and free energy retain their usual meanings.
\end{definition}

We work on a fixed state domain $X\subset\mathbb{R}^d$ over $t\in[t_0,T]$. The state fields are $(p(x,t),\theta(x,t),s(x,t))$, with auxiliary port fields specified as needed. The density $p$ is therefore treated not only as an endpoint distribution but as the state variable of a thermodynamically accounted probability path. The starting point of the path-dependent entropic Lagrangian construction is an energy-valued functional
\begin{equation}\label{eq:ledger-functional}
\mathcal A(t)
=
\int_X\big(\phi(p,\theta,\ldots)+\theta s\big)\,dx
+\int_X\int_{t_0}^{t} D(\tau)\,d\tau\,dx
+\mathcal P_{\partial X}(t),
\end{equation}
where $\phi$ is a Helmholtz-type free-energy density and the bulk pairing term $\theta s$ is included as part of the stored energy. The information/chemical potential is defined by
\begin{equation}\label{eq:mu-def-ledger}
\mu:=\partial_p\phi,
\end{equation}
with the understanding that if $\phi$ depends on spatial gradients or nonlocal operators, $\mu$ is interpreted as the corresponding variational derivative with respect to $p$. The history term collects energy-valued channel integrands through $D$, while $\mathcal P_{\partial X}(t)$ denotes explicit boundary-port power terms that are not carried implicitly by divergence channels. The upper-limit rule used below means that, under an admissible terminal-time generator, the history integral contributes its current channel integrand evaluated at $t$, after which each term is routed to its assigned channel. In the basic setting used throughout the paper,
\begin{equation}\label{eq:D-decomposition}
D=\nabla\cdot(\mu j)+\nabla\cdot q-D_{\mathrm{nd}}+D_{\mathrm{aux}},
\end{equation}
where $\nabla\cdot(\mu j)$ is the diffusion power port, $\nabla\cdot q$ is the heat port with the sign convention in Eq.~\ref{eq:fourier-q}, $D_{\mathrm{nd}}\ge 0$ denotes an optional non-divergence production density routed to the entropy channel with the sign shown in Eq.~\ref{eq:D-decomposition}, and $D_{\mathrm{aux}}$ denotes optional source, sink, observation, or boundary-exchange ports. Unless stated otherwise, $D_{\mathrm{nd}}=0$ and $D_{\mathrm{aux}}=0$.

The accounting variation is restricted to an admissible two-generator class, and no free test functions are introduced. Two local scalar generators $\delta t_1(x,t)$ and $\delta t_2(x,t)$ are postulated, representing, respectively, the entropy/irreversible accounting channel and the balance/kinematic accounting channel. Variations of state fields are induced only through these generators, namely
\begin{equation}\label{eq:admissible-generators}
\delta_{t_2}p:=\dot p\,\delta t_2,\qquad
\delta_{t_2}\theta:=\dot\theta\,\delta t_2,\qquad
\delta_{t_1}s:=\dot s\,\delta t_1,
\end{equation}
and no independent $\delta p$, $\delta\theta$, or $\delta s$ are permitted beyond Eq.~\ref{eq:admissible-generators}. The role of $\delta t_2$ is to generate balance-compatible variations consistent with the kinematics, whereas $\delta t_1$ generates entropy accounting; cross-variations and mixed $\delta t_1\delta t_2$ terms are excluded by construction. For a chemical-potential-weighted balance port, we use the induced weighted generator
\begin{equation}\label{eq:weighted-balance-generator}
\delta\eta_\mu:=\mu\,\delta t_2.
\end{equation}
Thus a balance-channel term of the form $\mu B\,\delta t_2$ is read equivalently as $B\,\delta\eta_\mu$. The family $\delta\eta_\mu$ is the effective balance test family associated with the port; where a pointwise degeneracy of $\mu$ is present, the balance closure is understood through smooth localization or regularization of this weighted test family.
 
Divergence channels are treated by explicit port-routing axioms that split each divergence port into a balance-channel closure in $\delta t_2$ and an entropy-channel contribution in $\delta t_1$. The diffusion divergence port is routed by the following axiom, stated as part of the calculus:
\begin{equation}\label{eq:diffusion-routing}
\delta \left[\int_X\int_{t_0}^{t}\nabla\cdot(\mu j)\,d\tau\,dx\right]
:=
\int_X\left( j\cdot\nabla\mu\,\delta t_1
+
\mu\,\nabla\cdot j\, \delta t_2\right)\,dx.
\end{equation}
The intended interpretation is that $j\cdot\nabla\mu$ contributes to entropy accounting, whereas $\mu\,\nabla\cdot j$ contributes to balance closure. In the weighted notation Eq.~\ref{eq:weighted-balance-generator}, the balance contribution is $\int_X (\nabla\cdot j)\,\delta\eta_\mu\,dx$. Heat transport is a pure entropy channel: the heat divergence port $\nabla\cdot q$ is routed exclusively into the $\delta t_1$ channel, with zero $\delta t_2$ contribution, so that it pairs directly with the $\delta t_1$-variation of the bulk term $\theta s$. With the convention $q=-k\nabla\theta$, the local entropy-channel equation is written as $\theta\dot s+\nabla\cdot q=\Xi$, equivalently $\theta\dot s=-\nabla\cdot q+\Xi$. Here $\Xi$ denotes the entropy-production term in the energy-valued accounting convention of the entropic Lagrangian. 

After applying the admissible generators Eq.~\ref{eq:admissible-generators}, and the routing axioms such as Eq.~\ref{eq:diffusion-routing}, the total accounting variation is collected in the canonical form
\begin{equation}\label{eq:canonical-variation}
\delta\mathcal A
=
\int_X R_1\,\delta t_1\,dx
+
\int_X R_2\,\delta t_2\,dx
+
\delta(\text{boundary ports}),
\end{equation}
where $\delta(\text{boundary ports})$ collects the variations of explicit boundary power terms from $\mathcal P_{\partial X}(t)$ and any boundary contributions induced by divergence channels when appropriate. The admissible coefficient-vanishing rule then imposes
\begin{equation}\label{eq:coeff-vanishing}
R_1=0,\qquad R_2=0,
\end{equation}
for arbitrary admissible generators $\delta t_1$ and $\delta t_2$, or equivalently for the induced weighted balance generator $\delta\eta_\mu$ on a $\mu$-weighted balance port. This yields, in subsequent sections, the variationally induced state relation $s=-\partial_\theta\phi$, the balance closure including the continuity equation for $p$, and the entropy equation with explicitly identified entropy-production terms. The present construction is an admissible-variation formalism for open irreversible systems with explicit port accounting and routing, rather than a free G\^ateaux variation or a reduction to action stationarity \cite{onsager1931reciprocal,degroot1984non}.

\section{State relation, continuity, and entropy accounting}\label{sec:ledger-induced}

\begin{proposition}\label{prop:ledger-induced}
Let $X\subset\mathbb{R}^d$ be fixed and let the bulk stored-energy density be $\phi(p,\theta,\ldots)+\theta s$, with $\mu:=\partial_p\phi$. Assume the admissible accounting generators satisfy $\delta_{t_2}p=\dot p\,\delta t_2$, $\delta_{t_2}\theta=\dot\theta\,\delta t_2$, and $\delta_{t_1}s=\dot s\,\delta t_1$, with no independent variations beyond these generators. Assume the diffusion divergence port is routed by
\begin{equation}\label{eq:diff-routing-sec4}
\delta \left[\int_X\int_{t_0}^{t}\nabla\cdot(\mu j)\,d\tau\,dx\right]
:=
\int_X\left( j\cdot\nabla\mu\,\delta t_1
+
 \mu\,\nabla\cdot j\,\delta t_2 \right)\,dx,
\end{equation}
and the heat divergence port is routed exclusively into the entropy channel, i.e.,
\begin{equation}\label{eq:heat-routing-sec4}
\delta \left[\int_X\int_{t_0}^{t}\nabla\cdot q\,d\tau\,dx\right]
:=\int_X \nabla\cdot q\,\delta t_1\,dx.
\end{equation}
Optional non-divergence production densities are routed by
\begin{equation}\label{eq:dnd-routing-sec4}
\delta \left[-\int_X\int_{t_0}^{t}D_{\mathrm{nd}}\,d\tau\,dx\right]
:=-\int_X D_{\mathrm{nd}}\,\delta t_1\,dx,\qquad D_{\mathrm{nd}}\ge 0.
\end{equation}
In the purely divergence-routed case, one sets $D_{\mathrm{nd}}=0$.
Then the coefficient-vanishing rule yields the state relation
\begin{equation}\label{eq:state-pairing}
{\ s=-\partial_\theta\phi\ },
\end{equation}
the balance closure containing the continuity equation
\begin{equation}\label{eq:continuity-sec4}
{\ \dot p+\nabla\cdot j=0\ },
\end{equation}
and the entropy equation in the form
\begin{equation}\label{eq:entropy-ledger-sec4}
{\ \theta\dot s+\nabla\cdot q=-\,j\cdot\nabla\mu+D_{\mathrm{nd}}\ }.
\end{equation}
Moreover, under the minimal constitutive closure
\begin{equation}\label{eq:onsager-closure-sec4}
{\ j=-M\nabla\mu,\qquad M\succeq 0\ },
\end{equation}
the entropy production density satisfies
\begin{equation}\label{eq:entropy-production-sec4}
{\ \Xi=\nabla\mu\cdot M\nabla\mu+D_{\mathrm{nd}}\ge 0\ }.
\end{equation}
\end{proposition}

\begin{proof}
By admissible balance accounting, the bulk variation of the stored energy reads
\begin{equation}\label{eq:bulk-var-sec4}
\delta_{t_2}\int_X\big(\phi(p,\theta,\ldots)+\theta s\big)\,dx
=
\int_X\Big(\partial_p\phi\,\dot p+\partial_\theta\phi\,\dot\theta+s\,\dot\theta\Big)\delta t_2\,dx
=
\int_X\Big(\mu\,\dot p+(\partial_\theta\phi+s)\dot\theta\Big)\delta t_2\,dx,
\end{equation}
since $\delta_{t_2}s=0$ and $\delta_{t_2}\theta=\dot\theta\,\delta t_2$. The latter identity is precisely the admissible thermal variation, so the thermal part may be read as $(\partial_\theta\phi+s)\,\delta\theta$ with $\delta\theta:=\delta_{t_2}\theta$. Coefficient vanishing with respect to this admissible $\delta\theta$ yields Eq.~\ref{eq:state-pairing}. 

The diffusion divergence port is treated by the routing axiom Eq.~\ref{eq:diff-routing-sec4}. Its balance-channel contribution is $\int_X \mu\,\nabla\cdot j\,\delta t_2\,dx$, which combines with the term $\int_X \mu\,\dot p\,\delta t_2\,dx$ from Eq.~\ref{eq:bulk-var-sec4} to form
\begin{equation}\label{eq:mu-balance-sec4}
\int_X \mu\big(\dot p+\nabla\cdot j\big)\,\delta t_2\,dx
=
\int_X \big(\dot p+\nabla\cdot j\big)\,\delta\eta_\mu\,dx.
\end{equation}
Here $\delta\eta_\mu:=\mu\delta t_2$ is the effective $\mu$-weighted balance generator introduced in Eq.~\ref{eq:weighted-balance-generator}. Coefficient vanishing is therefore applied to the weighted balance test in Eq.~\ref{eq:mu-balance-sec4}; by localization, and by regularization on possible zero sets of $\mu$ when needed, this gives the weak balance closure $\dot p+\nabla\cdot j=0$, namely Eq.~\ref{eq:continuity-sec4}. 

The entropy-channel variation of the bulk term arises from $\delta_{t_1}s=\dot s\,\delta t_1$, giving
\begin{equation}\label{eq:bulk-entropy-sec4}
\delta_{t_1}\int_X\big(\phi(p,\theta,\ldots)+\theta s\big)\,dx=\int_X \theta\,\dot s\,\delta t_1\,dx.
\end{equation}
By the heat routing rule Eq.~\ref{eq:heat-routing-sec4}, the heat divergence port contributes $\int_X \nabla\cdot q\,\delta t_1\,dx$. By the diffusion routing axiom Eq.~\ref{eq:diff-routing-sec4}, the entropy-channel contribution of the diffusion port is $\int_X j\cdot\nabla\mu\,\delta t_1\,dx$. The optional non-divergence production channel Eq.~\ref{eq:dnd-routing-sec4} contributes $-\int_X D_{\mathrm{nd}}\,\delta t_1\,dx$. Collecting the $\delta t_1$-terms and imposing coefficient vanishing yields
\begin{equation}\label{eq:R1-zero-sec4}
\theta\dot s+\nabla\cdot q+j\cdot\nabla\mu-D_{\mathrm{nd}}=0,
\end{equation}
which is equivalent to Eq.~\ref{eq:entropy-ledger-sec4}. 

Under the closure Eq.~\ref{eq:onsager-closure-sec4}, one has $-j\cdot\nabla\mu=\nabla\mu\cdot M\nabla\mu\ge 0$, hence Eq.~\ref{eq:entropy-production-sec4} follows with $\Xi:=\nabla\mu\cdot M\nabla\mu+D_{\mathrm{nd}}$. Such mobility closures are standard in nonequilibrium thermodynamics and dissipative evolution formulations \cite{onsager1931reciprocal,degroot1984non,jordan1998variational}.
\end{proof}

\section{KL/Shannon baseline example}\label{sec:kl-baseline}

We specify a baseline entropic free-energy contribution by fixing a reference density (base measure) $\pi(x)>0$ and defining
\begin{equation}\label{eq:phi-ent-kl}
\phi_{\mathrm{ent}}(p,\theta)=\theta\,p\log\frac{p}{\pi}.
\end{equation}
The associated entropic potential is obtained by differentiation with respect to $p$,
\begin{equation}\label{eq:mu-ent-kl}
\mu_{\mathrm{ent}}=\partial_p\phi_{\mathrm{ent}}
=\theta\Big(1+\log\frac{p}{\pi}\Big).
\end{equation}
The integral $\int_X p\log(p/\pi)\,dx$ is the relative entropy (Kullback--Leibler divergence) density functional, and its discrete counterpart reduces to Shannon's entropy up to the choice of base measure and units \cite{shannon1948,kullback1951}.

A common choice for the full free-energy density is an additive decomposition $\phi=\phi_0+\phi_{\mathrm{ent}}$, where $\phi_0$ encodes energetic preferences beyond entropic effects. For instance, an external potential $U(x)$ can be represented by $\phi_0(p)=p\,U(x)$, which yields $\partial_p\phi_0=U$ and hence
\begin{equation}\label{eq:mu-total-kl}
\mu=\partial_p\phi
=
U+\theta\Big(1+\log\frac{p}{\pi}\Big).
\end{equation}
Inserting Eq.~\ref{eq:mu-total-kl} into the balance closure $\partial_t p+\nabla\cdot j=0$ derived from the functional, one obtains a probabilistic flow once a mobility is specified. We adopt the Onsager-type constitutive closure
\begin{equation}\label{eq:mobility-kl}
j=-\mathbb{M}(p)\,\nabla\mu,
\qquad
\mathbb{M}(p)\succeq 0,
\end{equation}
which is consistent with nonnegative entropy production since $-j\cdot\nabla\mu=\nabla\mu\cdot \mathbb{M}(p)\nabla\mu\ge 0$. A standard probabilistic mobility is $\mathbb{M}(p)=p\,M$ with a constant positive semidefinite matrix $M\succeq 0$, which yields the Fokker--Planck-type evolution
\begin{equation}\label{eq:fp-kl}
\partial_t p
=
\nabla\cdot\Big(p\,M\nabla\mu\Big)
=
\nabla\cdot\Big(p\,M\nabla U+\theta\,p\,M\nabla\log\frac{p}{\pi}\Big),
\end{equation}
posed together with a no-flux boundary condition $j\cdot n=0$ on $\partial X$ to ensure conservation of total probability. The structure Eq.~\ref{eq:fp-kl} is a classical drift--diffusion form in probabilistic dynamics, and its relation to free-energy dissipation has been extensively studied \cite{risken1989fokker,jordan1998variational}.

Stationary states under no-flux boundary conditions satisfy $j=0$ and therefore $\nabla\mu=0$ in $X$. In the isothermal case with spatially uniform $\theta$, Eq.~\ref{eq:mu-total-kl} implies $\mu=\mu_\ast$ for a constant $\mu_\ast$, so that
\begin{equation}\label{eq:equilibrium-kl}
p(x)
=
Z^{-1}\,\pi(x)\,e^{-U(x)/\theta},
\qquad
Z:=\int_X \pi(x)\,e^{-U(x)/\theta}\,dx,
\end{equation}
provided $Z<\infty$. This equilibrium is the Gibbs form relative to the base measure $\pi$, and it reduces to the standard canonical density when $\pi$ is uniform.

With $D_{\mathrm{nd}}=0$, the entropy expression derived from the admissible calculus takes the form $\theta\dot s+\nabla\cdot q=\Xi$, with entropy production density $\Xi=-j\cdot\nabla\mu$. Under Eq.~\ref{eq:mobility-kl}, this becomes
\begin{equation}\label{eq:xi-kl}
\Xi=\nabla\mu\cdot\mathbb{M}(p)\nabla\mu\ge 0,
\end{equation}
which makes the dissipative character of the KL baseline explicit through a quadratic form in $\nabla\mu$. In this sense, the KL/Shannon choice Eq.~\ref{eq:phi-ent-kl} constitutes a baseline instance within the same structure, recovering a classical drift--diffusion setting while preserving the balance--entropy routing and the structural nonnegativity of entropy production \cite{degroot1984non,onsager1931reciprocal}.

\begin{remark}\label{rem:gauge-kl}
The additive term $\theta$ in Eq.~\ref{eq:mu-ent-kl} is a gauge in the isothermal setting, since only $\nabla\mu$ enters Eq.~\ref{eq:mobility-kl}. Equivalently, one may replace Eq.~\ref{eq:phi-ent-kl} by $\theta\,p(\log(p/\pi)-1)$, which yields $\mu_{\mathrm{ent}}=\theta\log(p/\pi)$ and leaves Eq.~\ref{eq:fp-kl} unchanged when $\theta$ is spatially uniform.
\end{remark}

\section{Applications}

\subsection{Application I: heavy tails and sparsity}\label{sec:app-heavy-tail}

The KL/Shannon choice $\mu_{\mathrm{ent}}=\theta\big(1+\log(p/\pi)\big)$ represents a single point within a broader class of admissible entropic potentials that preserve the balance--entropy routing and the structural nonnegativity of entropy production. In this section we introduce a minimal family of composable entropic potentials designed to adjust tail penalization and sparsity while keeping the same variational calculus and port accounting. Let $\pi(x)>0$ be a fixed reference density and define the relative density $u:=p/\pi$. We prescribe an entropic potential in the form
\begin{equation}\label{eq:mu-ent-general}
\mu_{\mathrm{ent}}=\theta\,\psi(u),
\end{equation}
where $\psi:(0,\infty)\to\mathbb{R}$ is a chosen scalar potential. Two representative choices are the power potential
\begin{equation}\label{eq:psi-alpha}
\psi_\alpha(u)=\frac{u^\alpha-1}{\alpha},\qquad \alpha\neq 0,
\end{equation}
and the $q$-logarithmic (Tsallis) potential
\begin{equation}\label{eq:psi-q}
\psi_q(u)=\ln_q u:=\frac{u^{1-q}-1}{1-q},\qquad q\neq 1,
\end{equation}
which is recovered from Eq.~\ref{eq:psi-alpha} by the identification $\alpha=1-q$. The limit $q\to 1$ yields $\psi_q(u)\to \log u$ and hence recovers the KL/Shannon baseline \cite{tsallis1988,naudts2011}. The adjective ``composable'' refers to the fact that $\psi_q$ satisfies a generalized additivity under multiplicative composition of relative densities, namely
\begin{equation}\label{eq:q-composition}
\psi_q(u_1u_2)=\psi_q(u_1)+\psi_q(u_2)+(1-q)\psi_q(u_1)\psi_q(u_2),
\end{equation}
which reduces to ordinary additivity in the limit $q\to 1$.

To embed Eq.~\ref{eq:mu-ent-general} into the framework, we define $\phi_{\mathrm{ent}}$ by integrating $\mu_{\mathrm{ent}}=\partial_p\phi_{\mathrm{ent}}$ with respect to $p$ at fixed $\theta$ and $\pi$. Writing $p=\pi u$ and $dp=\pi\,du$, one obtains
\begin{equation}\label{eq:phi-ent-general}
\phi_{\mathrm{ent}}(p,\theta)=\theta\int_{0}^{p}\psi \left(\frac{z}{\pi}\right)\,dz
=\theta\,\pi\int_{0}^{u}\psi(\xi)\,d\xi,
\end{equation}
up to an additive gauge depending only on $\theta$, which does not affect $\nabla\mu$ in isothermal settings. For the power potential Eq.~\ref{eq:psi-alpha}, Eq.~\ref{eq:phi-ent-general} yields the explicit density
\begin{equation}\label{eq:phi-ent-alpha}
\phi_{\mathrm{ent}}^{(\alpha)}(p,\theta)
=\theta\,\pi\left(\frac{u^{\alpha+1}}{\alpha(\alpha+1)}-\frac{u}{\alpha}\right),
\qquad u=\frac{p}{\pi},
\end{equation}
provided $\alpha>-1$ so that the primitive is finite at $u=0$. For the $q$-logarithmic potential Eq.~\ref{eq:psi-q}, one obtains analogously
\begin{equation}\label{eq:phi-ent-q}
\phi_{\mathrm{ent}}^{(q)}(p,\theta)
=\theta\,\pi\left(\frac{u^{2-q}}{(1-q)(2-q)}-\frac{u}{1-q}\right),
\qquad u=\frac{p}{\pi},
\end{equation}
for $q\neq 1,2$, with the limiting cases recovering the KL/Shannon form as $q\to 1$ \cite{naudts2011}.

With a decomposition $\phi=\phi_0+\phi_{\mathrm{ent}}$, the total potential reads
\begin{equation}\label{eq:mu-total-heavy}
\mu=\partial_p\phi=\partial_p\phi_0+\theta\,\psi(u),
\qquad u=\frac{p}{\pi}.
\end{equation}
With $D_{\mathrm{nd}}=0$, the admissible calculus yields the continuity equation $\partial_t p+\nabla\cdot j=0$ and the entropy equation $\theta\dot s+\nabla\cdot q=-j\cdot\nabla\mu$. Under the minimal mobility closure
\begin{equation}\label{eq:mobility-heavy}
j=-\mathbb{M}(p)\nabla\mu,
\qquad \mathbb{M}(p)\succeq 0,
\end{equation}
the induced probability flow takes the generic form
\begin{equation}\label{eq:pde-heavy}
\partial_t p=\nabla\cdot\big(\mathbb{M}(p)\nabla\mu\big),
\end{equation}
and the entropy production is $\Xi=\nabla\mu\cdot\mathbb{M}(p)\nabla\mu+D_{\mathrm{nd}}\ge 0$ by construction. If $\phi_0(p)=p\,U(x)$ for an external potential $U$, then $\partial_p\phi_0=U$ and stationary no-flux states satisfy $\nabla\mu=0$, hence
\begin{equation}\label{eq:stationary-inversion}
U(x)+\theta\,\psi \left(\frac{p(x)}{\pi(x)}\right)=\mu_\ast,
\end{equation}
for a constant $\mu_\ast$, which can be rewritten as
\begin{equation}\label{eq:stationary-solution-general}
p(x)=\pi(x)\,\psi^{-1}\left(\frac{\mu_\ast-U(x)}{\theta}\right),
\end{equation}
whenever $\psi$ is invertible on the relevant range. If $\psi$ has a finite lower endpoint, the inverse is understood on its active range and the density is set to zero when the argument falls below that endpoint, producing a finite cutoff. For the $q$-logarithmic choice Eq.~\ref{eq:psi-q}, Eq.~\ref{eq:stationary-solution-general} becomes the $q$-exponential form
\begin{equation}\label{eq:q-exponential-stationary}
p(x)=\pi(x)\left[1+(1-q)\frac{\mu_\ast-U(x)}{\theta}\right]_+^{\frac{1}{1-q}},
\end{equation}
where $[\cdot]_+$ denotes the positive part. The tail and sparsity properties follow directly from the asymptotic behavior of $\psi$ and $\psi^{-1}$. When $q>1$, the bracket in Eq.~\ref{eq:q-exponential-stationary} remains positive for large $U$ and the decay is algebraic, yielding heavy-tailed stationary states; when $q<1$, the positive-part constraint enforces compact support, leading to sparse (truncated) stationary densities. These regimes are well documented in nonextensive thermostatistics and in the associated nonlinear Fokker--Planck dynamics \cite{tsallis1988,borland1998,naudts2011}.

A concise criterion can be stated at the level of growth rates. Suppose $U(x)\to\infty$ as $|x|\to\infty$ and $\pi$ is slowly varying. The far-field regime in Eq.~\ref{eq:stationary-solution-general} corresponds to $(\mu_\ast-U(x))/\theta\to-\infty$ and therefore probes $\psi^{-1}(y)$ as $y\to-\infty$, equivalently the behavior of $\psi(u)$ as $u\to 0^+$. If $\psi(u)\sim\log u$ as $u\to0^+$, the Gibbs tail is recovered. If $|\psi(u)|/|\log u|\to\infty$ as $u\to0^+$, then $\psi^{-1}(y)$ decays more slowly than $e^y$ as $y\to-\infty$ and heavier-than-Gibbs tails result. If $|\psi(u)|/|\log u|\to0$, the inverse decays more rapidly and lighter tails result. If $\psi$ has a finite lower bound as $u\to0^+$, the inverse terminates at the lower endpoint of its active range and Eq.~\ref{eq:stationary-solution-general} yields compact support or sparse truncated profiles. The large-density behavior $u\to\infty$ instead controls the response to concentration and extreme upper ratios, hence it is more directly tied to robustness of transient driving forces than to the far-field tail under a confining $U$. In all cases, the modification of tail and sparsity behavior is achieved without altering the structure, the balance closure, or the nonnegativity of entropy production ensured by Eq.~\ref{eq:mobility-heavy}.

\subsection{Application II: bounded driving forces and robust potentials}\label{sec:app-robust}

The KL/Shannon baseline corresponds to an entropic potential of logarithmic growth, $\mu_{\mathrm{ent}}\sim \theta\log(p/\pi)$, and therefore yields an unbounded driving force as $p/\pi\to 0$ or $p/\pi\to\infty$. In applications where extreme density ratios arise, such as stiff transient regimes, truncated computational domains, or data-driven probability flows, it is often desirable to enforce boundedness of the entropic driving force while preserving the balance--entropy routing, the energy-valued accounting, and the structural nonnegativity of entropy production. This can be achieved by replacing the logarithmic potential by saturating (robust) composable potentials of the form $\mu_{\mathrm{ent}}=\theta\,\psi(u)$ with $u=p/\pi$, where $\psi$ is bounded and monotone.

Representative saturating choices include
\begin{equation}\label{eq:psi-saturating}
\psi_{\tanh}(u)=\frac{1}{\alpha}\tanh \big(\alpha\log u\big),\,\,
\psi_{\arctan}(u)=\frac{1}{\alpha}\arctan \big(\alpha\log u\big),\,\,
\psi_{\mathrm{ss}}(u)=\frac{\log u}{\sqrt{1+\alpha^2(\log u)^2}},
\end{equation}
with parameter $\alpha>0$. Each function in Eq.~\ref{eq:psi-saturating} satisfies $\psi(1)=0$ and is odd in $\log u$, hence it treats $u\to 0$ and $u\to\infty$ symmetrically in the logarithmic scale. Moreover, $\psi_{\tanh}$ and $\psi_{\arctan}$ are globally bounded, while $\psi_{\mathrm{ss}}$ is bounded with $\lim_{|\log u|\to\infty}\psi_{\mathrm{ss}}(u)=1/\alpha$. The corresponding entropic free-energy contribution can be constructed by integrating $\mu_{\mathrm{ent}}=\partial_p\phi_{\mathrm{ent}}$ as in Eq.~\ref{eq:phi-ent-general}, namely
\begin{equation}\label{eq:phi-ent-robust}
\phi_{\mathrm{ent}}(p,\theta)=\theta\int_{0}^{p}\psi \left(\frac{z}{\pi}\right)\,dz
=\theta\,\pi\int_{0}^{u}\psi(\xi)\,d\xi,
\qquad u=\frac{p}{\pi},
\end{equation}
up to an additive gauge depending only on $\theta$. The admissible variational calculus remains: it yields the continuity equation $\partial_t p+\nabla\cdot j=0$ and the entropy equation $\theta\dot s+\nabla\cdot q=-j\cdot\nabla\mu+D_{\mathrm{nd}}$, while the mobility closure $j=-\mathbb{M}(p)\nabla\mu$ with $\mathbb{M}(p)\succeq 0$ guarantees $\Xi=\nabla\mu\cdot\mathbb{M}(p)\nabla\mu+D_{\mathrm{nd}}\ge 0$.

\begin{lemma}\label{lem:bounded-driving}
Let $\mu_{\mathrm{ent}}=\theta\,\psi(u)$ with $u=p/\pi$, and assume $\theta$ is bounded on $X\times[t_0,T]$ and $\pi$ is strictly positive with bounded logarithmic gradient $\nabla\log\pi\in L^\infty(X)$. If $\psi$ is bounded on $(0,\infty)$, then $\mu_{\mathrm{ent}}$ is uniformly bounded whenever $\theta$ is. If, in addition, the logarithmic-scale derivative $u\psi'(u)$ is bounded on $(0,\infty)$ and $p$ is strictly positive, then $\nabla\mu_{\mathrm{ent}}$ admits the bound
\begin{equation}\label{eq:grad-bound}
|\nabla\mu_{\mathrm{ent}}|
\le
|\nabla\theta|\,|\psi(u)|
+
\|\theta\|_{L^\infty}\,\|u\psi'(u)\|_{L^\infty}\,|\nabla\log u|,
\qquad
\nabla\log u=\nabla\log p-\nabla\log\pi.
\end{equation}
Under the mobility closure $j=-\mathbb{M}(p)\nabla\mu$ with $\mathbb{M}(p)\succeq 0$, the entropic dissipation contribution satisfies
\begin{equation}\label{eq:xi-bound}
\nabla\mu\cdot\mathbb{M}(p)\nabla\mu \le \|\mathbb{M}(p)\|\,|\nabla\mu|^2,
\end{equation}
and therefore does not diverge due to unbounded growth of the entropic potential itself in the extreme-ratio regimes $u\to 0$ or $u\to\infty$ when $\psi$ and $u\psi'(u)$ are bounded.
\end{lemma}

\begin{proof}
The boundedness of $\mu_{\mathrm{ent}}$ follows directly from $\mu_{\mathrm{ent}}=\theta\,\psi(u)$ and $\|\psi\|_{L^\infty}<\infty$. Differentiating $\mu_{\mathrm{ent}}$ yields
\[
\nabla\mu_{\mathrm{ent}}=\psi(u)\nabla\theta+\theta\,\psi'(u)\nabla u
=\psi(u)\nabla\theta+\theta\,u\psi'(u)\nabla\log u.
\]
This gives Eq.~\ref{eq:grad-bound}. The estimate Eq.~\ref{eq:xi-bound} follows from the positive semidefiniteness of $\mathbb{M}(p)$.
\end{proof}

The contrast with the KL baseline is immediate. For $\psi(u)=\log u$, $\psi$ is unbounded on $(0,\infty)$ and the entropic potential $\mu_{\mathrm{ent}}$ diverges as $u\to 0$ or $u\to\infty$. In the robust choices Eq.~\ref{eq:psi-saturating}, $\psi$ is globally bounded and the logarithmic-scale derivative $u\psi'(u)$ is bounded, typically decaying to zero as $|\log u|\to\infty$. Thus extreme relative densities do not generate arbitrarily large entropic driving forces through growth of the potential scale. This modification changes only the entropic ``scale'' while preserving the structure, the balance closure, and the entropy-production mechanism. Such bounded or tempered potentials are consistent with the general philosophy of robust divergences and tempered likelihoods in statistical modeling, where extreme ratios are intentionally down-weighted to improve stability \cite{basu1998robust,grunwald2017safe}.

\subsection{Application III: shape control via structural energies}\label{sec:app-shape-control}

The entropic choice $\phi_{\mathrm{ent}}(p,\theta)$ controls pointwise penalties in the density variable, but it does not by itself prescribe geometric regularity of the density profile. Shape control can be incorporated within the same variational calculus by placing structural terms in the non-entropic part $\phi_0$ of the free energy, allowing $\phi=\phi_0+\phi_{\mathrm{ent}}$. In this section we consider three representative mechanisms: Fisher-type smoothing, total-variation edge preservation, and curvature penalization that suppresses high-frequency oscillations beyond simple smoothing. Throughout, the information/chemical potential is understood as the variational derivative $\mu=\delta\phi/\delta p$ when $\phi$ depends on spatial derivatives of $p$.

A Fisher-type structural energy can be introduced by
\begin{equation}\label{eq:phi0-fisher}
\phi_0(p,\nabla p)=\kappa\,|\nabla \sqrt{p}|^2,
\qquad \kappa>0,
\end{equation}
which is equivalent to $\phi_0=\frac{\kappa}{4}\,|\nabla p|^2/p$ for $p>0$. Writing $\rho:=\sqrt{p}$, one has $\phi_0=\kappa|\nabla\rho|^2$ and $p=\rho^2$, hence a variation $\delta p=2\rho\,\delta\rho$ yields
\begin{equation}\label{eq:mu-fisher}
\delta\int_X \kappa|\nabla\rho|^2\,dx
=
-2\kappa\int_X (\Delta\rho)\,\delta\rho\,dx
=
-\kappa\int_X \frac{\Delta\sqrt{p}}{\sqrt{p}}\,\delta p\,dx,
\end{equation}
so that the corresponding contribution to the potential is
\begin{equation}\label{eq:mu0-fisher}
\mu_0^{\mathrm{F}}=\frac{\delta}{\delta p}\int_X \phi_0\,dx
=
-\kappa\,\frac{\Delta\sqrt{p}}{\sqrt{p}}.
\end{equation}
Under the mobility closure $j=-\mathbb{M}(p)\nabla\mu$ with $\mathbb{M}(p)\succeq 0$, this adds a Laplacian-type regularizing component to the driving force and hence promotes smoothing of the density profile while preserving the conservative form $\partial_t p+\nabla\cdot j=0$. Fisher-type regularizations and their variational structures are classical in information geometry and density functional formulations \cite{villani2003topics,carlen1991}.

Edge-preserving and piecewise-constant profiles can be promoted by a total-variation structural energy
\begin{equation}\label{eq:phi0-tv}
\phi_0(p,\nabla p)=\kappa\,|\nabla p|,
\qquad \kappa>0,
\end{equation}
which is not differentiable at $\nabla p=0$ and is therefore typically implemented via a smooth approximation
\begin{equation}\label{eq:phi0-tv-eps}
\phi_0^{\varepsilon}(p,\nabla p)=\kappa\,\sqrt{|\nabla p|^2+\varepsilon^2},
\qquad \varepsilon>0.
\end{equation}
The corresponding variational derivative is obtained by integration by parts:
\begin{equation}\label{eq:mu-tv-eps}
\delta\int_X \kappa\sqrt{|\nabla p|^2+\varepsilon^2}\,dx
=
\int_X \kappa\,\frac{\nabla p}{\sqrt{|\nabla p|^2+\varepsilon^2}}\cdot\nabla(\delta p)\,dx
=
-\int_X \kappa\,\nabla\cdot \left(\frac{\nabla p}{\sqrt{|\nabla p|^2+\varepsilon^2}}\right)\delta p\,dx,
\end{equation}
hence
\begin{equation}\label{eq:mu0-tv}
\mu_0^{\mathrm{TV},\varepsilon}
=
-\kappa\,\nabla\cdot \left(\frac{\nabla p}{\sqrt{|\nabla p|^2+\varepsilon^2}}\right),
\end{equation}
with the limiting subdifferential interpretation as $\varepsilon\to 0$ recovering the TV potential. This term penalizes total variation rather than squared gradients, thereby suppressing diffusion across sharp transitions and promoting edge preservation in the evolving density profile, a mechanism widely used in variational imaging and related PDE flows \cite{rudin1992nonlinear,chambolle2004}.

To suppress high-frequency oscillations beyond first-order smoothing, one may include a curvature-type structural energy that penalizes second derivatives of $\log p$. A representative choice is
\begin{equation}\label{eq:phi0-curvature}
\phi_0(p,\nabla^2 p)=\kappa\,p\,\|\nabla^2\log p\|^2,
\qquad \kappa>0,
\end{equation}
where $\|\cdot\|$ denotes the Frobenius norm. This term is invariant under multiplicative rescalings of $p$ and directly controls curvature of the log-density, thereby penalizing rapid oscillations and stabilizing multimodal shapes without enforcing excessive diffusion of edges. Denoting $\ell:=\log p$, one has $\nabla^2\ell=\nabla^2 p/p-\nabla p\otimes\nabla p/p^2$, and the variational derivative $\mu_0^{\mathrm{C}}=\delta\int_X \phi_0\,dx/\delta p$ is well-defined on smooth strictly positive densities and yields a higher-order (typically fourth-order) contribution to the driving force. While the explicit expression is lengthy, its structure follows from repeated integration by parts and can be treated as a variational derivative operator acting on $\ell$; in particular, it produces a dissipative regularization that targets curvature rather than merely gradients, analogous in spirit to higher-order variational regularizers used to suppress oscillations while preserving salient features \cite{chan2000higher}.

\subsection{Application IV: nonlocal correlations and multimodality}\label{sec:app-nonlocal}

Pointwise entropic potentials and local structural energies can regulate tails and smoothness, but they do not directly encode long-range correlations or collective effects that lead to clustering, repulsion, and multimodal stationary profiles. Such phenomena can be represented within the same variational calculus by introducing a nonlocal interaction energy in the non-entropic component $\phi_0$. Let $\mathcal K:X\to\mathbb{R}$ be an interaction kernel and denote the convolution (or more generally the integral operator) by
\begin{equation}\label{eq:convolution}
(\mathcal K * p)(x):=\int_X \mathcal K(x-y)\,p(y)\,dy,
\end{equation}
with the understanding that Eq.~\ref{eq:convolution} can be replaced by $(\mathcal K p)(x)=\int_X \mathcal K(x,y)p(y)\,dy$ on general domains. We introduce the nonlocal interaction energy density
\begin{equation}\label{eq:phi0-nonlocal}
\phi_0(p,\mathcal K*p)=\frac12\,p\,(\mathcal K*p),
\end{equation}
so that the associated nonlocal contribution to the potential is the variational derivative
\begin{equation}\label{eq:mu-nonlocal}
\mu_{\mathrm{nl}}=\frac{\delta}{\delta p}\left[\frac12\int_X p(\mathcal K*p)\,dx\right]=\mathcal K*p,
\end{equation}
assuming $\mathcal K$ is symmetric so that the quadratic form is well-defined. For a nonsymmetric kernel, the variational derivative is given by the symmetrized operator. The full potential becomes $\mu=\partial_p\phi_{\mathrm{ent}}+\mu_{\mathrm{loc}}+\mu_{\mathrm{nl}}$, where $\mu_{\mathrm{loc}}$ denotes any local structural contributions such as those in Section~\ref{sec:app-shape-control}. Under the mobility closure $j=-\mathbb{M}(p)\nabla\mu$ with $\mathbb{M}(p)\succeq 0$, the probability flow reads
\begin{equation}\label{eq:pde-nonlocal}
\partial_t p=\nabla\cdot\big(\mathbb{M}(p)\nabla\mu\big)
=
\nabla\cdot\big(\mathbb{M}(p)\nabla(\partial_p\phi_{\mathrm{ent}}+\mu_{\mathrm{loc}}+\mathcal K*p)\big),
\end{equation}
and the entropy production retains the structural form $\Xi=\nabla\mu\cdot\mathbb{M}(p)\nabla\mu+D_{\mathrm{nd}}\ge 0$. Nonlocal interaction energies of the form Eq.~\ref{eq:phi0-nonlocal} are standard in aggregation--diffusion models and nonlocal gradient flows, where they serve as a principled mechanism for inducing clustering or repulsion depending on the sign structure of the kernel \cite{carrillo2003asymptotic,carrillo2010particle}.

Stationary states under no-flux boundary conditions satisfy $j=0$ and therefore $\nabla\mu=0$ in $X$, implying that $\mu$ is spatially constant on each connected component. In the isothermal setting with $\phi_{\mathrm{ent}}$ fixed and with an optional external potential $U$ included through $\phi_0(p)=pU$, the stationarity condition takes the general form
\begin{equation}\label{eq:stationary-nonlocal}
\partial_p\phi_{\mathrm{ent}}(p,\theta)+U+\mu_{\mathrm{loc}}(p)+\mathcal K*p=\mu_\ast,
\end{equation}
for a constant $\mu_\ast$. The qualitative structure of stationary solutions is governed by the competition between the entropic/local regularizing terms and the nonlocal interaction term. If $\mathcal K$ is predominantly attractive, then $\mathcal K*p$ tends to reinforce concentration of mass and can induce clustering and multimodality; if $\mathcal K$ is predominantly repulsive or positive-definite in an appropriate sense, then $\mathcal K*p$ tends to spread mass and can enforce separation between peaks or suppress aggregation. A useful diagnostic is obtained by linearizing Eq.~\ref{eq:pde-nonlocal} about a homogeneous state on periodic domains: in Fourier variables, the interaction contribution introduces a multiplier $\widehat{\mathcal K}(k)$. A mode becomes unstable only when this interaction contribution overcomes the stabilizing local entropic and structural terms in the linearized coefficient. Thus attractive kernels with sufficiently negative $\widehat{\mathcal K}(k)$ over a band of modes can trigger pattern-forming instabilities, whereas repulsive or positive-definite kernels tend to stabilize homogeneous profiles \cite{carrillo2010particle}.

Combining the nonlocal interaction Eq.~\ref{eq:phi0-nonlocal} with the structural energies in Section~\ref{sec:app-shape-control} yields controlled multimodality. Fisher-type or curvature-type terms penalize excessive oscillations and provide higher-order regularization, while TV-type terms can preserve sharp interfaces between clustered regions. Consequently, the framework supports multimodal probability flows in which clustering is induced by nonlocal interactions and high-frequency artifacts are controlled by local structural energies, without altering the balance closure, the entropy equation, or the structural nonnegativity of entropy production.

\begin{table}[t]
\centering
\scriptsize
\setlength{\tabcolsep}{3pt}
\renewcommand{\arraystretch}{1.12}
\caption{Summary of modular choices within the same balance--entropy routing.}
\label{tab:module-summary}
\begin{tabular}{p{0.18\textwidth}p{0.27\textwidth}p{0.27\textwidth}p{0.20\textwidth}}
\toprule
Module & Potential/component & Effect and stationary behavior & Notes/assumptions \\
\midrule
KL/Shannon baseline & $\mu_{\mathrm{ent}}=\theta(1+\log(p/\pi))$ or gauge-shifted $\theta\log(p/\pi)$ & Reference drift--diffusion; Gibbs recovery $p\propto \pi e^{-U/\theta}$ under no-flux stationarity & $p>0$, $Z<\infty$; additive gauge allowed \\
Composable $q$-log / power scale & $\mu_{\mathrm{ent}}=\theta\psi_q(p/\pi)$ with $\psi_q=\ln_q$ & Tunable tail and sparsity control; algebraic heavy tails for $q>1$ and compact support for $q<1$ & Monotone invertible scale on the active range \\
Saturating robust scale & $\mu_{\mathrm{ent}}=\theta f(\log(p/\pi))$ with bounded $f$ & Bounded potential drive under extreme density ratios; tempered or truncated stationary profiles & Bounded $f$ and bounded log-scale derivative $f'$ \\
Fisher / TV / curvature structure & $\phi_0=\kappa|\nabla\sqrt p|^2$, $\kappa|\nabla p|$, or $\kappa p\|\nabla^2\log p\|^2$ & Smoothing, edge preservation, or suppression of high-frequency oscillations; higher-order or nonsmooth stationarity & Boundary conditions and regularity; TV via subgradient or $\varepsilon$-regularization \\
Nonlocal interaction & $\phi_0=\frac12 p(\mathcal K*p)$ & Correlations, clustering, repulsion, and multimodality; integral stationarity condition with $\mathcal K*p$ & Symmetric kernel for the displayed derivative; stability depends on local terms and $\widehat{\mathcal K}$ \\
\bottomrule
\end{tabular}
\end{table}

\noindent This table is meant as a model-design map: the same balance closure and entropy routing are kept fixed, while the entropic chart or non-entropic structural component is substituted according to the desired effect.

\renewcommand{\arraystretch}{1.0}

\section{From maximum-entropy states to information-driven probability paths}\label{sec:maxent-information-paths}

\subsection{Probability laws as information states}\label{sec:probability-information-states}

Information theory assigns quantitative functionals to a probability law, while inference and learning change that law. Let $\mathcal P(X)$ denote the set of normalized densities on $X$. A static information principle selects an element $p^\ast\in\mathcal P(X)$, whereas an information-processing dynamics generates a path
\begin{equation}\label{eq:probability-path-map}
[t_0,T]\ni t\longmapsto p_t\in\mathcal P(X).
\end{equation}
The distinction is fundamental. Entropy, relative entropy, and mutual information quantify properties or relations of probability states; probability transport, Bayesian assimilation, and generative evolution require rules for moving between such states. The entropic Lagrangian supplies this second layer by combining a probability-balance channel, an entropy-production channel, and external information ports within one history-aware variational object.

For the KL/Shannon sector, a convenient time-dependent free-energy functional is
\begin{equation}\label{eq:information-free-energy}
\mathcal F_t[p]
=
\int_X\left[
\theta p\log\frac{p}{\pi}+p\,U(x,t)
\right]dx,
\end{equation}
where $\pi>0$ is a reference density, $\theta>0$ is an information/thermal scale, and $U$ is an information potential. Depending on the application, $U$ may encode physical energy, moment constraints, a negative log-likelihood, observations, or other externally supplied information. Its variational drive is
\begin{equation}\label{eq:information-mu}
\mu=\frac{\delta\mathcal F_t}{\delta p}
=U(x,t)+\theta\left(1+\log\frac{p}{\pi}\right).
\end{equation}
With $j=-\mathbb M(p)\nabla\mu$, the same object determines both the selected equilibrium state and the path by which it is approached.

\subsection{Maximum entropy as the stationary closed-system sector}\label{sec:maxent-to-pdel}

The maximum-entropy (MaxEnt) principle selects a static probability density by maximizing Shannon entropy, or equivalently minimizing relative entropy, subject to constraints that encode the available information \cite{shannon1948,jaynes1957,cover2006}. Given constraint functions $\{f_k\}_{k=1}^m$ and prescribed moments $\{c_k\}_{k=1}^m$, the relative-entropy problem reads
\begin{equation}\label{eq:maxent-kl}
\min_{p\ge 0}\ \int_X p(x)\log\frac{p(x)}{\pi(x)}\,dx
\quad\text{subject to}\quad
\int_X p\,dx=1,\qquad
\int_X p\,f_k\,dx=c_k,\ \ k=1,\dots,m.
\end{equation}
Introducing Lagrange multipliers $\lambda_0,\lambda_1,\dots,\lambda_m$, one considers
\begin{equation}\label{eq:maxent-lagrangian}
\mathcal L[p]
=
\int_X p\log\frac{p}{\pi}\,dx
+\lambda_0\!\left(\int_X p\,dx-1\right)
+\sum_{k=1}^m \lambda_k\!\left(\int_X p\,f_k\,dx-c_k\right).
\end{equation}
The stationarity condition is
\begin{equation}\label{eq:maxent-eulerlagrange}
\log\frac{p(x)}{\pi(x)}+1+\lambda_0+\sum_{k=1}^m \lambda_k f_k(x)=0,
\end{equation}
and therefore
\begin{equation}\label{eq:maxent-exponential}
p^\ast(x)
=
Z^{-1}\,\pi(x)\exp\!\left(-\sum_{k=1}^m \lambda_k f_k(x)\right),
\qquad
Z=\int_X \pi(x)\exp\!\left(-\sum_{k=1}^m \lambda_k f_k(x)\right)dx.
\end{equation}
A single constraint $f_1=U$ with $\lambda_1=\beta$ gives the Gibbs form $p^\ast\propto\pi e^{-\beta U}$. This is a state-selection rule: it identifies $p^\ast$ but does not itself prescribe a probability path, an entropy-production identity, or external information ports.

The same state arises dynamically by choosing
\begin{equation}\label{eq:pdel-entropic-choice}
\phi_{\mathrm{ent}}(p,\theta)=\theta p\log\frac{p}{\pi},
\qquad
\phi_0(p)=p\,U(x),
\end{equation}
which gives
\begin{equation}\label{eq:pdel-mu-maxent}
\mu(x,t)=U(x)+\theta\left(1+\log\frac{p(x,t)}{\pi(x)}\right).
\end{equation}
Under the mobility closure
\begin{equation}\label{eq:pdel-mobility-maxent}
j=-\mathbb M(p)\nabla\mu,
\qquad
\mathbb M(p)\succeq0,
\end{equation}
the continuity equation becomes $\partial_t p=\nabla\cdot(\mathbb M(p)\nabla\mu)$ and the entropy-production term is nonnegative. For a connected domain, a positive stationary state with vanishing flux and nondegenerate mobility has $\nabla\mu=0$, hence
\begin{equation}\label{eq:pdel-stationary-maxent}
U(x)+\theta\left(1+\log\frac{p^\ast(x)}{\pi(x)}\right)=\mu_\ast,
\end{equation}
which yields
\begin{equation}\label{eq:pdel-gibbs-recovery}
p^\ast(x)=Z^{-1}\,\pi(x)\exp\!\left(-\frac{U(x)}{\theta}\right),
\qquad
Z=\int_X\pi(x)\exp\!\left(-\frac{U(x)}{\theta}\right)dx.
\end{equation}

\begin{proposition}[Maximum entropy as the stationary sector]\label{prop:maxent-stationary-sector}
Let $X$ be connected, $\theta>0$, $\pi>0$, and let $U=\theta\sum_{k=1}^m\lambda_k f_k$. Assume the mobility is positive definite on the support of a positive normalized density. Then the MaxEnt solution Eq.~\ref{eq:maxent-exponential} is a zero-flux stationary state of the probability flow generated by Eqs.~\ref{eq:pdel-mu-maxent}--\ref{eq:pdel-mobility-maxent}. Conversely, every positive zero-flux stationary state of this flow satisfies the MaxEnt stationarity condition for the corresponding constraints.
\end{proposition}

\begin{proof}
Equation~\ref{eq:maxent-eulerlagrange} and the definition of $U$ imply
$\mu=U+\theta(1+\log(p^\ast/\pi))=-\theta\lambda_0$, which is spatially constant. Hence $\nabla\mu=0$, $j=0$, and the density is stationary. Conversely, $j=0$ and positive-definite mobility imply $\nabla\mu=0$. Since $X$ is connected, $\mu$ is constant; solving Eq.~\ref{eq:pdel-stationary-maxent} and enforcing normalization gives Eq.~\ref{eq:maxent-exponential}.
\end{proof}

The preceding equivalence is not restricted to a finite-dimensional exponential family. It also yields a precise representation result for general probability laws.

\begin{corollary}[Representation of a prescribed probability law]\label{cor:probability-representation}
Let $p^\dagger$ be any strictly positive normalized density on $X$ and let $\pi>0$. For any $\theta>0$, define the information potential
\begin{equation}\label{eq:target-information-potential}
U^\dagger(x):=-\theta\log\frac{p^\dagger(x)}{\pi(x)}+C,
\end{equation}
where $C$ is arbitrary. Then $p^\dagger$ is the stationary no-flux state generated by the KL/Shannon entropic module and the potential $U^\dagger$.
\end{corollary}

\begin{proof}
Substitution of Eq.~\ref{eq:target-information-potential} into Eq.~\ref{eq:pdel-stationary-maxent} gives a spatially constant chemical potential. Equation~\ref{eq:pdel-gibbs-recovery} then returns $p^\dagger$ after normalization.
\end{proof}

Corollary~\ref{cor:probability-representation} formalizes the broad representational reach of maximum entropy: any positive probability law can be expressed as a Gibbs/MaxEnt state relative to a suitable reference and information potential. With finitely many prescribed moments, the admissible potentials are restricted to the span of the selected constraint functions; with a general potential, the representation is correspondingly general. The substantive modeling problem is therefore not whether a target law can be represented, but how its information potential, constraints, and dynamical pathway are determined.

\subsection{Information injection and the pathwise free-energy identity}\label{sec:information-port-identity}

A time-dependent information potential produces an open probability-path problem. We assume fixed $\theta$ and $\pi$, a no-flux boundary, and a time-differentiable potential $U(x,t)$. Define the internal dissipation and supplied information power by
\begin{equation}\label{eq:dissipation-information-power}
\mathcal D(t):=\int_X\nabla\mu\cdot\mathbb M(p)\nabla\mu\,dx,
\qquad
\mathcal P_{\mathrm{info}}(t):=\int_X p(x,t)\,\partial_tU(x,t)\,dx.
\end{equation}

\begin{proposition}[Dissipation--information decomposition]\label{prop:information-power-identity}
For the probability flow generated by Eqs.~\ref{eq:information-free-energy}--\ref{eq:information-mu} and $j=-\mathbb M(p)\nabla\mu$, the free-energy rate satisfies
\begin{equation}\label{eq:information-power-identity}
\frac{d}{dt}\mathcal F_t[p_t]
=
-\mathcal D(t)+\mathcal P_{\mathrm{info}}(t).
\end{equation}
\end{proposition}

\begin{proof}
Using the chain rule, the continuity equation, and the no-flux boundary condition gives
\begin{align}
\frac{d}{dt}\mathcal F_t[p_t]
&=\int_X\mu\,\partial_tp\,dx+\int_Xp\,\partial_tU\,dx \\
&=-\int_X\mu\,\nabla\cdot j\,dx+\mathcal P_{\mathrm{info}}(t) \\
&=\int_X\nabla\mu\cdot j\,dx+\mathcal P_{\mathrm{info}}(t) \\
&=-\int_X\nabla\mu\cdot\mathbb M(p)\nabla\mu\,dx+\mathcal P_{\mathrm{info}}(t),
\end{align}
which is Eq.~\ref{eq:information-power-identity}.
\end{proof}

Equation~\ref{eq:information-power-identity} separates two conceptually different mechanisms. When no new information enters, $\partial_tU=0$ and the probability path relaxes monotonically through internal dissipation. When observations, data, controls, or environmental signals change $U$, the term $\mathcal P_{\mathrm{info}}$ measures their instantaneous contribution to the energy-valued information ledger. Thus information updating is represented as an open-system port, rather than being hidden inside an endpoint optimization.

\subsection{Bayesian inference as an information-port realization}\label{sec:bayes-to-pdel}

Bayesian inference updates a prior density $p_0(x)$ into a posterior density $p(x\mid y)$ after observing data $y$, using a likelihood $L(x):=p(y\mid x)$:
\begin{equation}\label{eq:bayes-rule}
p(x\mid y)=\frac{p_0(x)L(x)}{\int_Xp_0(z)L(z)\,dz}.
\end{equation}
The posterior is equivalently the unique minimizer
\begin{equation}\label{eq:bayes-variational}
p^\ast=\arg\min_{p\in\mathcal P(X)}\left\{\mathrm{KL}(p\|p_0)-\int_Xp(x)\log L(x)\,dx\right\}.
\end{equation}
With a multiplier $\lambda$ for normalization, the corresponding functional is
\begin{equation}\label{eq:bayes-lagrangian}
\mathcal J[p]=\int_Xp\log\frac{p}{p_0}\,dx-\int_Xp\log L\,dx+\lambda\left(\int_Xp\,dx-1\right),
\end{equation}
and its first-order condition
\begin{equation}\label{eq:bayes-eulerlagrange}
\log\frac{p(x)}{p_0(x)}+1-\log L(x)+\lambda=0
\end{equation}
recovers Bayes' rule. For conditionally independent observations $y_1,\dots,y_N$, the posterior is $p_N(x)\propto p_0(x)\prod_{n=1}^NL_n(x)$ \cite{mackay2003,bernardo2009}.

To realize this update dynamically, choose $\pi\equiv1$ and introduce the observation potential
\begin{equation}\label{eq:bayes-U}
U(x,t):=-\theta\log p_0(x)-\theta\sum_{n=1}^{N(t)}\log L_n(x),
\end{equation}
with the KL entropic density
\begin{equation}\label{eq:pdel-bayes-ent}
\phi_{\mathrm{ent}}(p,\theta)=\theta p\log p.
\end{equation}
The induced potential is
\begin{equation}\label{eq:pdel-bayes-mu}
\mu=U(x,t)+\theta(1+\log p).
\end{equation}
A stationary no-flux state therefore satisfies
\begin{equation}\label{eq:pdel-bayes-stationary}
p(x,t)=Z(t)^{-1}\exp\!\left(-\frac{U(x,t)}{\theta}\right),
\end{equation}
and substitution of Eq.~\ref{eq:bayes-U} gives
\begin{equation}\label{eq:pdel-bayes-posterior}
p(x,t)=Z(t)^{-1}p_0(x)\prod_{n=1}^{N(t)}L_n(x),
\end{equation}
which is precisely the Bayesian posterior after the information available at time $t$ has been assimilated. The evidence is the normalization factor $Z(t)$.

The dynamical formulation adds information that is absent from the endpoint rule. Between observation events, $U$ is fixed and Eq.~\ref{eq:information-power-identity} gives monotone internal relaxation. During a smooth assimilation schedule, $\mathcal P_{\mathrm{info}}=\int_Xp\,\partial_tU\,dx$ measures the rate at which observation information enters the system. For stepwise data, the same interpretation applies piecewise: each change of $U$ is an information-port event, followed by dissipative relaxation toward the updated posterior. The balance channel continues to conserve total probability, while the entropy channel accounts for the irreversible part of the update.

For a Gaussian illustration, let $p_0(x)=\mathcal N(0,\sigma_0^2)$ and $L(x)=\mathcal N(y;x,\sigma^2)$. The posterior is
\begin{equation}\label{eq:gauss-posterior}
p(x\mid y)=\mathcal N(m,\sigma_{\mathrm{post}}^2),
\qquad
\sigma_{\mathrm{post}}^2=\left(\sigma_0^{-2}+\sigma^{-2}\right)^{-1},
\qquad
m=\sigma_{\mathrm{post}}^2\frac{y}{\sigma^2}.
\end{equation}
The potential $U=-\theta\log p_0-\theta\log L$ is quadratic, and Eq.~\ref{eq:pdel-bayes-stationary} reproduces the same Gaussian mean and variance. Replacing the logarithmic entropic scale by a bounded or composable $\psi$ produces a tempered or robust posterior while preserving the same balance--entropy and information-port structure.

\subsection{Implications for information theory and probabilistic learning}\label{sec:information-ai-implications}

The mathematical objects of information theory are probability laws and functionals defined on them; a substantial part of modern learning theory and generative modeling likewise operates by optimizing, approximating, or transporting distributions. The present construction provides a common dynamic interpretation of these activities. Information functionals define the landscape, the chemical/information potential supplies the local drive, mobility determines the geometry and rate of adaptation, and external data enter through explicit ports. Table~\ref{tab:information-learning-interpretation} summarizes this correspondence.

\begin{table}[htbp]
\centering
\small
\caption{Interpretation of the probability-path variables in information and learning systems.}
\label{tab:information-learning-interpretation}
\begin{tabularx}{\textwidth}{p{0.20\textwidth}YY}
\toprule
Mathematical object & Probability interpretation & Information/learning interpretation \\
\midrule
$p(x,t)$ & Time-dependent probability law & Belief, approximate posterior, model distribution, or generated distribution \\
$\mathcal F_t[p]$ & Free-energy/information functional & Learning objective, regularized risk, or variational inference objective \\
$\mu=\delta\mathcal F_t/\delta p$ & Variational driving potential & Local pressure for probability reallocation \\
$\mathbb M(p)$ & Mobility of probability mass & Update geometry, adaptation rate, or preconditioning \\
$U(x,t)$ and auxiliary ports & External or structural potential & Data, likelihood, observations, controls, or environmental information \\
History channel & Accumulated expenditure & Irreversible learning cost and cumulative information processing \\
\bottomrule
\end{tabularx}
\end{table}

For variational inference, $p_t$ may be an evolving approximate posterior and $U$ collects prior and likelihood information. For diffusion and score-based generation, $p_t$ is the transported law between noise and data regimes, while the driving potential determines the probability flow \cite{sohl2015deep,ho2020denoising,song2021scorebased}. For sequential Bayesian or online inference, changes in $U$ represent newly arriving data and Eq.~\ref{eq:information-power-identity} separates information injection from internal dissipation. These examples differ in constitutive detail, but they share the same state variable---a probability law---and the same structural question: how should that law evolve under information, balance, and irreversibility constraints?

The framework therefore has a scope broader than any single diffusion equation. It supplies a thermodynamic variational architecture in which probability is the state, information potentials are driving forces, entropy production encodes irreversibility, history records accumulated expenditure, and ports represent interaction with data or the environment. Different information-theoretic and learning models can be constructed by changing the constitutive modules while retaining the same probability-balance and entropy-accounting structure.

\subsection{Discussion}\label{sec:discussion-maxent-bayes}

Maximum entropy and Bayesian inference select or update probability states; the entropic Lagrangian supplies an admissible path between them. MaxEnt is recovered as the stationary closed-system sector, whereas a time-dependent information potential acts as an external port and separates supplied information power from internal dissipation.

The central contribution is the path-dependent, channel-resolved construction: upper-limit history terms, balance closure, entropy accounting, and external ports remain explicit while the information or structural potential is changed. This extends state-based information principles to probability transport and open-system updating without reducing the process to endpoint optimization.

\section{Numerical illustrations and discrete energy-ledger checks}\label{sec:numerical}

Two finite-volume examples verify the structure rather than benchmark solvers: mass is conserved, the total free energy satisfies the discrete dissipation ledger, and its components identify the mechanisms that drive or resist the observed pattern. For grid cells $x_i$ with spacing $\Delta x$, density values $p_i(t)$, and potential values $\mu_i(t)$, define the edge flux
\begin{equation}\label{eq:num-fv-flux}
J_{i+1/2}=-m_{i+1/2}\frac{\mu_{i+1}-\mu_i}{\Delta x},
\qquad
\dot p_i=-\frac{J_{i+1/2}-J_{i-1/2}}{\Delta x},
\end{equation}
where $m_{i+1/2}\ge 0$ is the edge mobility. We use $m_{i+1/2}=(p_i+p_{i+1})/2$ in the periodic case and the analogous one-sided no-flux closure at a boundary. Summation of Eq.~\ref{eq:num-fv-flux} over the cells telescopes, so the discrete balance channel conserves total probability. The discrete dissipation channel is
\begin{equation}\label{eq:num-discrete-diss}
D_h(t)=\sum_i m_{i+1/2}
\left(\frac{\mu_{i+1}-\mu_i}{\Delta x}\right)^2\Delta x\ge0.
\end{equation}
For a discrete free energy $F_h[p]$ whose cell-wise variational derivative is $\mu_i$, the semi-discrete identity is
\begin{equation}\label{eq:num-total-energy-ledger}
\frac{dF_h}{dt}=-D_h,
\qquad
R_E(t):=F_h(t)-F_h(0)+\int_0^tD_h(\tau)\,d\tau=0.
\end{equation}
The computations evaluate the energy components, $D_h$, and its accumulated integral independently, so $R_E$ is a genuine total-energy residual.

\subsection{Example 1: composable entropic chart, stationary tails, and energy components}\label{sec:num-tail}

The first test changes only the entropic chart while keeping the same stationary rule. On $X=[-12,12]$, take a uniform reference density $\pi$, temperature $\theta=1$, and confining potential $U(x)=x^2/2$. For the $q$-logarithmic module
\begin{equation}\label{eq:num-q-log}
\ln_q u=\frac{u^{1-q}-1}{1-q},
\qquad \ln_1 u=\log u,
\end{equation}
the no-flux stationary equation is
\begin{equation}\label{eq:num-q-stationary}
U_i+\theta\ln_q\left(\frac{p_i}{\pi_i}\right)=C.
\end{equation}
The constant $C$ is chosen by the normalization $\sum_i p_i\Delta x=1$, using
\begin{equation}\label{eq:num-q-exp}
p_i(C)=\pi_i\exp_q\left(\frac{C-U_i}{\theta}\right),
\qquad
\exp_q(y)=\left[1+(1-q)y\right]_+^{1/(1-q)}.
\end{equation}

To display the energetic role of the constitutive substitution, let $\Psi_q'(u)=\ln_q u$ and $\Psi_q(1)=0$. The discrete stationary free energy is decomposed as
\begin{align}
F_{q,h}&=F_{U,h}+F_{\mathrm{ent},q,h},\label{eq:num-tail-energy-total}\\
F_{U,h}&=\sum_iU_i p_i\Delta x,\qquad
F_{\mathrm{ent},q,h}=\theta\sum_i\pi_i
\Psi_q\left(\frac{p_i}{\pi_i}\right)\Delta x.\label{eq:num-tail-energy-components}
\end{align}
For $q=1$, $\Psi_1(u)=u\log u-u+1$; for $q\ne1,2$,
\begin{equation}
\Psi_q(u)=
\frac{u^{2-q}-(2-q)u+(1-q)}{(1-q)(2-q)}.
\end{equation}
Figure~\ref{fig:num-tail-module} compares the KL case $q=1$, a heavy-tailed case $q=1.35$, and a compact-support case $q=0.65$. The first two panels show that the same balance and stationary-potential conditions generate Gaussian-type, algebraic-tail, or sparse profiles solely through the entropic chart. The third panel resolves $F_{U,h}$ and $F_{\mathrm{ent},q,h}$. These component values are interpreted within each constitutive choice; because changing $q$ changes the free-energy functional itself, the absolute totals are not used as a cross-$q$ ranking.

\begin{figure}[t]
\centering
\includegraphics[width=0.99\textwidth]{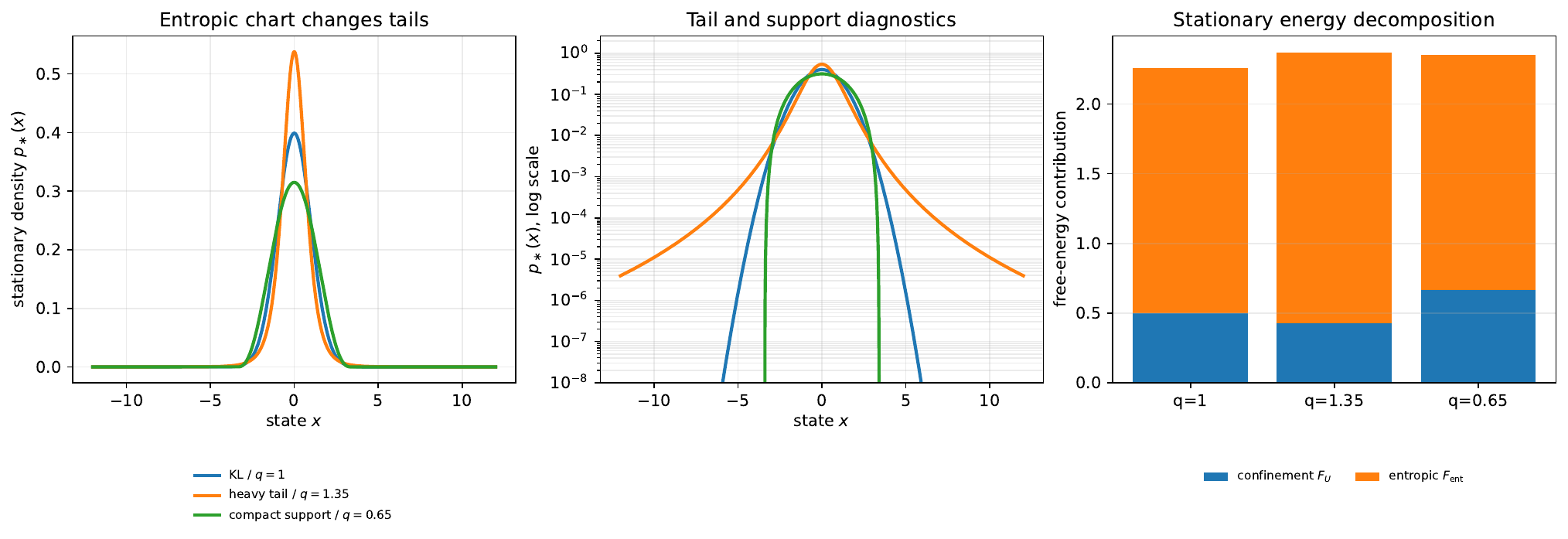}
\caption{Stationary profiles and energy components for the composable $q$-logarithmic potential under the same confinement and normalization. The left panel shows the densities, the middle panel resolves their tails and support on a logarithmic scale, and the right panel decomposes the stationary free energy into confinement and entropic contributions.}
\label{fig:num-tail-module}
\end{figure}
\FloatBarrier

\subsection{Example 2: nonlocal clustering and the resolved total-energy ledger}\label{sec:num-nonlocal}

The second test keeps the KL entropic scale and adds a nonlocal component to $\phi_0$. On the periodic interval $X=[0,2\pi]$, take
\begin{equation}\label{eq:num-nonlocal-mu}
\mu_i=\theta\log\left(\frac{p_i}{\pi_i}\right)+(\mathcal K_h p)_i,
\qquad
\mathcal K(r)=-A\cos(2r),
\end{equation}
with $\theta=0.06$, $A=0.55$, and uniform $\pi$. The kernel is attractive in the second Fourier mode; its interaction contribution overcomes the stabilizing entropic contribution in that mode and drives a two-cluster pattern. The total discrete free energy is resolved as
\begin{align}
F_h[p]&=F_{\mathrm{ent},h}[p]+F_{\mathrm{int},h}[p],\label{eq:num-nonlocal-energy}\\
F_{\mathrm{ent},h}[p]
&=\theta\sum_i p_i\log\left(\frac{p_i}{\pi_i}\right)\Delta x,\label{eq:num-nonlocal-entropic-energy}\\
F_{\mathrm{int},h}[p]
&=\frac12\sum_i p_i(\mathcal K_h p)_i\Delta x.\label{eq:num-nonlocal-interaction-energy}
\end{align}

The energy components identify the clustering mechanism. Over the run, $\Delta F_{\mathrm{ent},h}=+0.087719$, $\Delta F_{\mathrm{int},h}=-0.242575$, and $\Delta F_h=-0.154857$. The entropic term therefore resists concentration, whereas the attractive interaction drives it strongly enough to decrease the total free energy.

Figure~\ref{fig:num-nonlocal-clustering} shows the density, energy components, accumulated dissipation, and total-energy residual. The energy loss $F_h(0)-F_h(T)$ and the accumulated dissipation $\int_0^T D_h\,dt$ both equal $0.1548567414$. The final and maximum relative residuals are $8.1\times10^{-15}$ and $1.4\times10^{-14}$; the energy-partition residual is $2.8\times10^{-17}$.

\begin{figure}[t]
\centering
\includegraphics[width=0.99\textwidth]{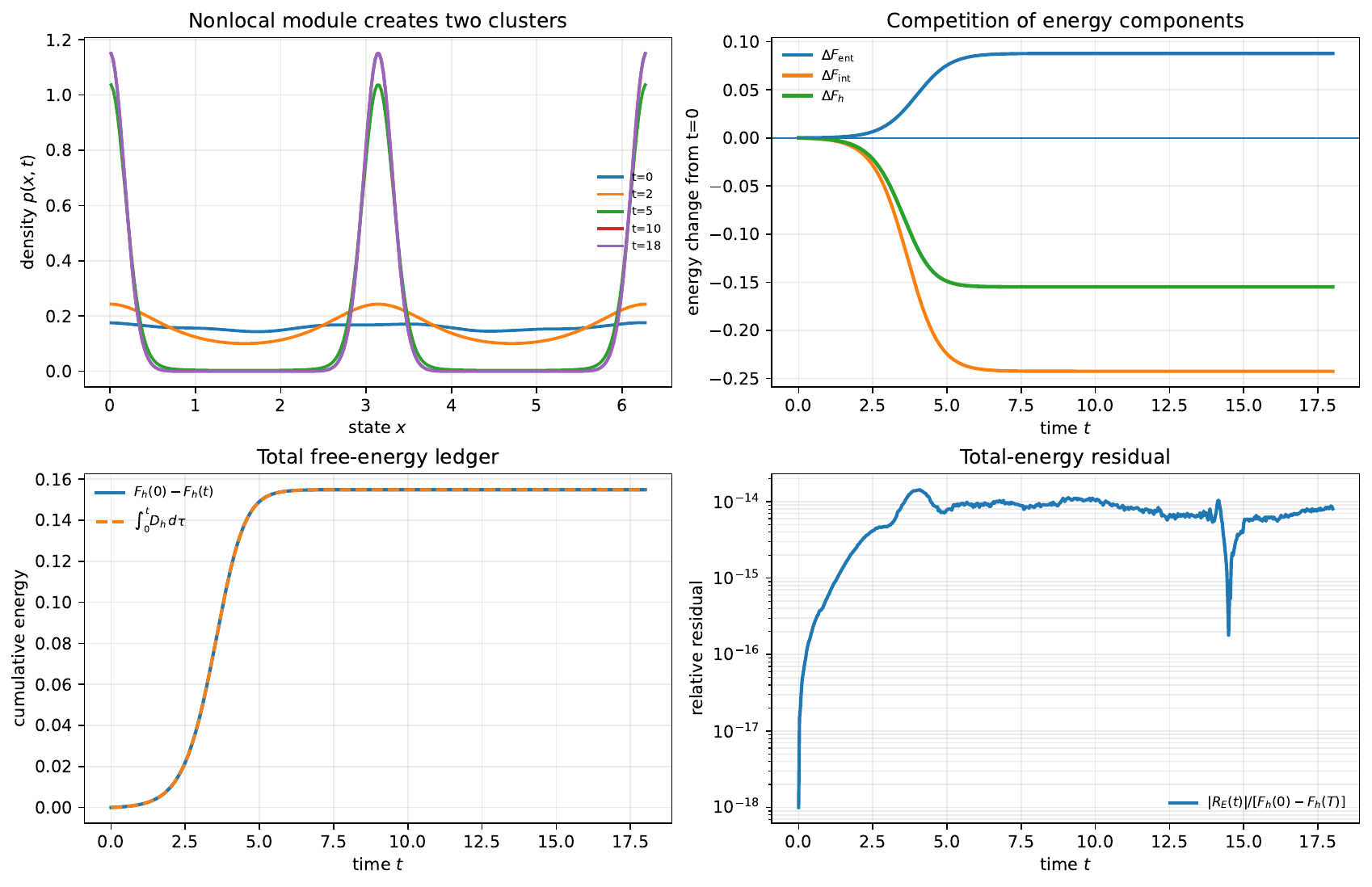}
\caption{Resolved energy ledger for the nonlocal KL flow. Top left: formation of the two-cluster profile. Top right: the entropic contribution increases and opposes concentration, while the attractive interaction contribution decreases more strongly; their sum is the decreasing total free energy. Bottom left: total-energy loss and independently accumulated dissipation coincide. Bottom right: the relative total-energy residual remains near machine precision.}
\label{fig:num-nonlocal-clustering}
\end{figure}
\FloatBarrier

\begin{table}[t]
\centering
\scriptsize
\setlength{\tabcolsep}{4pt}
\renewcommand{\arraystretch}{1.12}
\caption{Numerical verification of balance, energy components, dissipation, and the total-energy residual. The active-support fraction counts cells with $p_i>10^{-6}\max_jp_j$.}
\label{tab:numerical-validation}
\begin{tabular}{lll}
\toprule
Test & Quantity & Value \\
\midrule
Stationary tails, $q=1$ & mass error; max $|\mu-C|$ & $3.3\times10^{-16}$; $1.8\times10^{-15}$ \\
Stationary tails, $q=1.35$ & mass error; max $|\mu-C|$ & $2.2\times10^{-16}$; $2.4\times10^{-14}$ \\
Stationary tails, $q=0.65$ & mass error; max $|\mu-C|$ & $2.2\times10^{-16}$; $1.8\times10^{-15}$ \\
Stationary tails & support fractions, $q=1,1.35,0.65$ & $0.438$; $1.000$; $0.283$ \\
Stationary energy, $q=1$ & $F_U$; $F_{\rm ent}$; $F_q$ & $0.500000$; $1.759115$; $2.259115$ \\
Stationary energy, $q=1.35$ & $F_U$; $F_{\rm ent}$; $F_q$ & $0.427443$; $1.942422$; $2.369865$ \\
Stationary energy, $q=0.65$ & $F_U$; $F_{\rm ent}$; $F_q$ & $0.665548$; $1.686731$; $2.352280$ \\
Nonlocal clustering & maximum mass drift; minimum $D_h$ & $8.9\times10^{-16}$; $1.1\times10^{-16}$ \\
Nonlocal components & $\Delta F_{\rm ent}$; $\Delta F_{\rm int}$; $\Delta F_h$ & $+0.087719$; $-0.242575$; $-0.154857$ \\
Total-energy ledger & $F_h(0)-F_h(T)$; $\int_0^T D_h\,dt$ & $0.1548567414$; $0.1548567414$ \\
Total-energy residual & final relative; maximum relative & $8.1\times10^{-15}$; $1.4\times10^{-14}$ \\
Independent checks & energy-partition residual; energy increases & $2.8\times10^{-17}$; $0$ \\
\bottomrule
\end{tabular}
\end{table}
\FloatBarrier

\section{Conclusions and outlook}\label{sec:summary-scope}

This paper developed a path-dependent entropic Lagrangian calculus for probability flows. An energy-valued stored state, upper-limit history terms, restricted accounting generators, and explicit port routing jointly produce the thermal state relation, conservative probability balance, and nonnegative production under mobility closure.

Maximum entropy selects probability states, whereas the entropic Lagrangian generates admissible probability paths. The KL/Shannon sector recovers MaxEnt and Bayesian laws as stationary no-flux states; time-dependent information potentials yield
\[
\frac{d}{dt}\mathcal F_t[p_t]=-\mathcal D(t)+\mathcal P_{\mathrm{info}}(t),
\]
which separates internal relaxation from supplied information. Composable entropic charts and structural energies then change tails, robustness, regularity, and nonlocal organization without changing the balance--entropy architecture. The numerical examples confirm this modularity together with mass conservation and the resolved total-energy ledger.

Because probability laws underlie information theory, inference, and probabilistic learning, the formulation provides a common thermodynamic language for probability transport, information injection, irreversible production, and open-system exchange. Future work will address data assimilation, generative probability flows, interacting learning systems, and structure-preserving discretizations.

\appendix

\section{Composable potentials and induced composition laws}\label{app:composable}

Let $\pi:X\to(0,\infty)$ be a fixed reference density and set the relative density $u:=p/\pi$. We consider entropic potentials of the form $\mu_{\mathrm{ent}}=\theta\,\psi(u)$ with $\psi:(0,\infty)\to\mathbb{R}$, and we call $\psi$ composable if it is generated by a monotone chart on the logarithmic scale. Specifically, let $f:\mathbb{R}\to I\subset\mathbb{R}$ be strictly monotone with $f(0)=0$, and define
\begin{equation}\label{eq:appA-psi-f}
\psi(u):=f(\log u),\qquad u>0.
\end{equation}
For $u_1,u_2>0$ one has $\log(u_1u_2)=\log u_1+\log u_2$, and therefore
\begin{equation}\label{eq:appA-comp}
\psi(u_1u_2)=f(\log u_1+\log u_2)=:y_1\oplus_f y_2,
\qquad y_i:=\psi(u_i).
\end{equation}
If $f$ is invertible on its range, the induced composition on the $y$-variable is explicitly
\begin{equation}\label{eq:appA-oplus}
y_1\oplus_f y_2:=f \big(f^{-1}(y_1)+f^{-1}(y_2)\big),
\end{equation}
with identity element $0$ since $\psi(1)=f(0)=0$. The logarithmic baseline corresponds to $f(\xi)=\xi$, so that $\psi(u)=\log u$ and $\oplus_f$ reduces to ordinary addition. Nonlinear charts $f$ generate deformed additivity, which can be used to temper extreme ratios $u\to 0$ or $u\to\infty$ while maintaining a closed algebra under multiplicative composition of relative densities. 

Several choices of $f$ recover common generalized logarithms. If $f(\xi)=(e^{\alpha\xi}-1)/\alpha$ for $\alpha\neq 0$, then $\psi(u)=(u^\alpha-1)/\alpha$ and
\begin{equation}\label{eq:appA-alpha-oplus}
\psi(u_1u_2)=\psi(u_1)+\psi(u_2)+\alpha\,\psi(u_1)\psi(u_2).
\end{equation}
If $\alpha=1-q$, then $\psi=\ln_q$ is the $q$-logarithm and Eq.~\ref{eq:appA-alpha-oplus} becomes the standard $q$-composition rule \cite{tsallis1988,naudts2011}. If $f(\xi)=\tfrac{1}{\alpha}\tanh(\alpha\xi)$, then $\psi(u)=\tfrac{1}{\alpha}\tanh(\alpha\log u)$ is bounded, and the induced $\oplus_f$ is non-polynomial but remains associative and commutative by construction through Eq.~\ref{eq:appA-oplus}.

\section{Asymptotic criteria for heavy tails, compact support, and bounded driving}\label{app:asymptotics}

Assume $\theta$ is spatially uniform and $\phi=\phi_0+\phi_{\mathrm{ent}}$, with $\phi_0(p)=p\,U(x)$ for an external potential $U$ and $\mu_{\mathrm{ent}}=\theta\,\psi(p/\pi)$. Stationary no-flux states satisfy $\nabla\mu=0$ and hence
\begin{equation}\label{eq:appB-stationary}
U(x)+\theta\,\psi \left(\frac{p(x)}{\pi(x)}\right)=\mu_\ast,
\end{equation}
for a constant $\mu_\ast$. If $\psi$ is strictly increasing, Eq.~\ref{eq:appB-stationary} yields the inversion on the active range of $\psi$
\begin{equation}\label{eq:appB-inversion}
p(x)=\pi(x)\,\psi^{-1}\left(\frac{\mu_\ast-U(x)}{\theta}\right).
\end{equation}
The far-field regime $|x|\to\infty$ typically corresponds to $U(x)\to\infty$, which forces $(\mu_\ast-U(x))/\theta\to-\infty$ and therefore probes the asymptotics of $\psi^{-1}(y)$ as $y\to-\infty$, equivalently the behavior of $\psi(u)$ as $u\to 0^+$. If $\psi_-:=\lim_{u\to0^+}\psi(u)>-\infty$, Eq.~\ref{eq:appB-inversion} is understood only where $(\mu_\ast-U)/\theta>\psi_-$, and $p$ is set to zero outside this active range.

If $\psi(u)\sim \log u$ as $u\to 0^+$, then $\psi^{-1}(y)\sim e^{y}$ as $y\to-\infty$ and Eq.~\ref{eq:appB-inversion} recovers the Gibbs tail $p\sim \pi\,e^{-U/\theta}$. If $|\psi(u)|/|\log u|\to\infty$ as $u\to0^+$, then $\psi^{-1}(y)$ decays more slowly than $e^{y}$ as $y\to-\infty$, and Eq.~\ref{eq:appB-inversion} produces heavier tails relative to the Gibbs baseline. Conversely, if $|\psi(u)|/|\log u|\to0$, the inverse decays more rapidly and yields lighter tails. A particularly transparent family is $\psi_q(u)=\ln_q u$: for $q>1$, one has $\psi_q(u)\sim -u^{1-q}/(q-1)$ as $u\to 0^+$, so that $\psi_q^{-1}(y)\sim \big(-(q-1)y\big)^{-1/(q-1)}$ as $y\to-\infty$, yielding algebraic heavy tails; for $q<1$, $\ln_q u$ has a finite lower bound as $u\to 0^+$ and the active-range interpretation gives a cutoff, yielding compact support in Eq.~\ref{eq:appB-inversion} \cite{tsallis1988,borland1998,naudts2011}.

Robustness and bounded driving are controlled by the growth of $\psi(u)$ as $u\to 0^+$ and $u\to\infty$, together with its derivative on the logarithmic density-ratio scale. If $\psi$ is bounded on $(0,\infty)$, then for bounded $\theta$ the entropic potential $\mu_{\mathrm{ent}}=\theta\psi(u)$ is uniformly bounded, preventing divergence of the driving force due solely to extreme ratios. If $u\psi'(u)$ is bounded and $p$ remains strictly positive, then
\[
\nabla\mu_{\mathrm{ent}}=\psi(u)\nabla\theta+\theta\,u\psi'(u)\nabla\log u
\]
can be estimated in terms of $\nabla\log u$, and the quadratic dissipation $\nabla\mu\cdot\mathbb{M}(p)\nabla\mu$ cannot blow up due to growth of $\psi$ or its logarithmic-scale derivative as $u\to 0,\infty$.

\section{Variational derivatives for Fisher, TV, and curvature regularizers}\label{app:varder}

We record standard variational derivatives used in Sections~\ref{sec:app-shape-control} and \ref{sec:app-nonlocal}. All expressions are understood under boundary conditions compatible with the required integrations by parts, such as periodicity or no-flux-type conditions for the corresponding operators, and for smooth strictly positive densities $p$ when $\log p$ is involved.

For the Fisher-type term
\begin{equation}\label{eq:appC-fisher}
\mathcal F_{\mathrm{F}}[p]:=\int_X \kappa\,|\nabla\sqrt p|^2\,dx,
\qquad \kappa>0,
\end{equation}
writing $\rho=\sqrt p$ and using $\delta p=2\rho\,\delta\rho$ yields
\begin{equation}\label{eq:appC-fisher-var}
\delta \mathcal F_{\mathrm{F}}[p]
=
2\kappa\int_X \nabla\rho\cdot\nabla(\delta\rho)\,dx
=
-2\kappa\int_X (\Delta\rho)\,\delta\rho\,dx
=
-\kappa\int_X \frac{\Delta\sqrt p}{\sqrt p}\,\delta p\,dx,
\end{equation}
hence
\begin{equation}\label{eq:appC-mu-fisher}
\frac{\delta \mathcal F_{\mathrm{F}}}{\delta p}
=
-\kappa\,\frac{\Delta\sqrt p}{\sqrt p}.
\end{equation}

For the smoothed total-variation term
\begin{equation}\label{eq:appC-tv}
\mathcal F_{\mathrm{TV}}^\varepsilon[p]:=\int_X \kappa\,\sqrt{|\nabla p|^2+\varepsilon^2}\,dx,
\qquad \kappa>0,\ \varepsilon>0,
\end{equation}
one obtains
\begin{equation}\label{eq:appC-tv-var}
\delta \mathcal F_{\mathrm{TV}}^\varepsilon[p]
=
\kappa\int_X \frac{\nabla p}{\sqrt{|\nabla p|^2+\varepsilon^2}}\cdot\nabla(\delta p)\,dx
=
-\kappa\int_X \nabla\cdot \left(\frac{\nabla p}{\sqrt{|\nabla p|^2+\varepsilon^2}}\right)\delta p\,dx,
\end{equation}
hence
\begin{equation}\label{eq:appC-mu-tv}
\frac{\delta \mathcal F_{\mathrm{TV}}^\varepsilon}{\delta p}
=
-\kappa\,\nabla\cdot \left(\frac{\nabla p}{\sqrt{|\nabla p|^2+\varepsilon^2}}\right).
\end{equation}

For the curvature-type term
\begin{equation}\label{eq:appC-curv}
\mathcal F_{\mathrm{C}}[p]:=\int_X \kappa\,p\,\|\nabla^2\log p\|^2\,dx,
\qquad \kappa>0,
\end{equation}
let $\ell=\log p$, so that $\delta \ell=\delta p/p$ and $\delta(\nabla^2\ell)=\nabla^2(\delta p/p)$. A direct computation gives
\begin{equation}\label{eq:appC-curv-var}
\delta \mathcal F_{\mathrm{C}}[p]
=
\kappa\int_X \Big(\|\nabla^2\ell\|^2\,\delta p
+2p\,\nabla^2\ell:\nabla^2 \left(\frac{\delta p}{p}\right)\Big)\,dx,
\end{equation}
and repeated integration by parts yields a fourth-order operator acting on $\ell$ of the form
\begin{equation}\label{eq:appC-curv-mu-formal}
\frac{\delta \mathcal F_{\mathrm{C}}}{\delta p}
=
\kappa\Big(\|\nabla^2\log p\|^2+\mathcal L_{\mathrm{C}}[\log p]\Big),
\end{equation}
where $\mathcal L_{\mathrm{C}}$ denotes the resulting linear differential operator obtained by distributing the derivatives in the second term of Eq.~\ref{eq:appC-curv-var}. The explicit expanded form is lengthy and not required for the qualitative conclusions drawn in the main text; it can be generated routinely by symbolic differentiation and integration by parts. Higher-order variational regularizers of this type are standard in PDE-based regularization models \cite{chan2000higher}.

\section{Minimal analytic comparison in one dimension}\label{app:minimal-example}

We consider $X=\mathbb{R}$, $\pi\equiv 1$, constant temperature $\theta>0$, and an external potential $U(x)=\frac12 x^2$. We take the probabilistic mobility $\mathbb{M}(p)=p$ so that $j=-p\,\partial_x\mu$ and the balance closure reads $\partial_t p=\partial_x(p\,\partial_x\mu)$. 

For the KL baseline, $\phi_{\mathrm{ent}}=\theta\,p\log p$ and $\mu=U+\theta(1+\log p)$. The stationary no-flux condition $j=0$ implies $\partial_x\mu=0$, hence $\mu=\mu_\ast$ and
\begin{equation}\label{eq:appD-gauss}
p_\ast(x)=Z^{-1}\exp \left(-\frac{x^2}{2\theta}\right),
\qquad Z=\sqrt{2\pi\theta}.
\end{equation}
The dissipation term induced by the entropy production is $\int_{\mathbb{R}} p|\partial_x\mu|^2\,dx\ge 0$, and it vanishes at equilibrium.

For a bounded entropic potential $\mu_{\mathrm{ent}}=\theta\,\psi(p)$ with $\psi$ chosen as in Eq.~\ref{eq:psi-saturating}, the stationary condition becomes $U(x)+\theta\psi(p(x))=\mu_\ast$, hence, on the active range of $\psi$,
\begin{equation}\label{eq:appD-sat-inversion}
p_\ast(x)=\psi^{-1}\left(\frac{\mu_\ast-\tfrac12 x^2}{\theta}\right).
\end{equation}
If $\psi$ has a finite lower bound, the density is set to zero when the argument in Eq.~\ref{eq:appD-sat-inversion} falls below that bound, yielding compact support and sparsity; if $\psi(u)\to -\infty$ as $u\to 0^+$ and $|\psi(u)|/|\log u|\to\infty$, then $p_\ast$ exhibits heavier tails relative to Eq.~\ref{eq:appD-gauss}. In all cases, boundedness of $\psi$ and of the logarithmic-scale derivative $u\psi'(u)$ prevents arbitrarily large entropic driving forces arising from growth of the potential scale in transient regimes, and the entropy production retains the quadratic form $\int_{\mathbb{R}} p|\partial_x\mu|^2\,dx$.

Including the Fisher regularizer $\mathcal F_{\mathrm{F}}[p]=\int_{\mathbb{R}}\kappa|\partial_x\sqrt p|^2\,dx$ adds $\mu_0^{\mathrm{F}}=-\kappa(\partial_{xx}\sqrt p)/\sqrt p$ to the potential, so that
\begin{equation}\label{eq:appD-mu-fisher-1d}
\mu(x)=\frac12 x^2+\theta(1+\log p)-\kappa\,\frac{\partial_{xx}\sqrt p}{\sqrt p}.
\end{equation}
Stationarity $\partial_x\mu=0$ yields a second-order nonlinear boundary-value problem for $\sqrt p$ that balances confinement, entropic spreading, and Fisher smoothing. The corresponding entropy production density remains $\Xi=p|\partial_x\mu|^2\ge 0$ under the mobility choice, and the Fisher term selectively damps high-frequency components of $\sqrt p$, providing smoothing without modifying the balance--entropy routing.

\section*{Acknowledgments}
The research is supported by the Fundamental Research Funds for the Central Universities (Grant No.02302350113).

\bibliographystyle{apalike}
\bibliography{PDELinfo}

\end{document}